\documentclass{article}

\usepackage{amsmath}
\usepackage{amsthm}
\usepackage{thmtools}
\usepackage{amssymb}
\usepackage{bbm}
\usepackage{mathtools}
\usepackage{tikz}
\usepackage{MnSymbol}
\usepackage{lipsum,hyperref}
\usepackage{graphicx}
\usepackage{amsfonts}
\usepackage{xfrac}
\usepackage{authblk}
\usepackage{color,soul}
\usepackage{framed}
\usepackage{setspace}
\usepackage{dutchcal}
\usepackage{rotating}
\usepackage{pdflscape}
\usepackage{setspace}
\usepackage{titlesec}
\usepackage{caption}
\usepackage{xstring}
\usepackage[framemethod=TikZ]{mdframed}
\usepackage[a4paper, total={6in, 8in}]{geometry}

\newcommand{\memomode}{OFF} %

\graphicspath{{./drawings/}}

\hypersetup{colorlinks=true}	

\newcommand{\eqdef}{\vcentcolon=}

\newcommand{\set}[2]{\left\{ #1 \mathrel{}:\mathrel{} #2\right\}}

\newtheorem*{corollary*}{Corollary}
\newtheorem*{definition*}{Definition}

\newcommand{\genstirlingII}[3]{%
	\genfrac{\{}{\}}{0pt}{#1}{#2}{#3}%
}

\newcommand{\stirlingII}[2]{\genstirlingII{}{#1}{#2}}

\newtheorem{innercustomthm}{Theorem}
\newenvironment{customthm}[1]
{\renewcommand\theinnercustomthm{#1}\innercustomthm}
{\endinnercustomthm}

\newcommand\restr[2]{{\left.\kern-\nulldelimiterspace#1\vphantom{\big|}\right|_{#2}}}

\definecolor{nicegrey}{HTML}{e6e6ea}
\renewenvironment{leftbar}{%
	\MakeFramed {\advance\hsize-\width \FrameRestore}}%
{\endMakeFramed}

\newenvironment{thinleftbar}{%
	\MakeFramed {\advance\hsize-\width \FrameRestore}}%
{\endMakeFramed}

\definecolor{niceorange}{HTML}{fed766}
\newenvironment{theo}[2][]{%
	\refstepcounter{theo}
	\ifstrempty{#1}%
	{\mdfsetup{%
			frametitle={%
				\tikz[baseline=(current bounding box.east),outer sep=0pt]
				\node[anchor=east,rectangle,fill=blue!20]
				{\strut Theorem~\thetheo};}
		}%
	}{\mdfsetup{%
			frametitle={%
				\tikz[baseline=(current bounding box.east),outer sep=0pt, inner ysep=1pt]
				\node[anchor=east,rectangle, draw=nicegrey, line width=0.7pt, fill=white]
				{\strut Theorem~\thetheo:~#1};}%
		}%
	}%
	\mdfsetup{%
		innertopmargin=1pt,linecolor=nicegrey,%
		linewidth=0.7pt,topline=true,%
		frametitleaboveskip=\dimexpr-\ht\strutbox\relax%
	}
	\begin{mdframed}[]\relax}{%
\end{mdframed}}

\definecolor{niceblue}{HTML}{0e9aa7}
\newenvironment{openq}[2][]{%
	\refstepcounter{openq}
	\ifstrempty{#1}%
	{\mdfsetup{%
			frametitle={%
				\tikz[baseline=(current bounding box.east),outer sep=0pt]
				\node[anchor=east,rectangle,fill=niceblue!20]
				{\strut Open Question~\theopenq};}
		}%
	}{\mdfsetup{%
			frametitle={%
				\tikz[baseline=(current bounding box.east),outer sep=0pt, inner ysep=1pt]
				\node[anchor=east,rectangle, fill=niceblue!15]
				{\strut Open Question~\theopenq:~#1};}%
		}%
	}%
	\mdfsetup{%
		innertopmargin=1pt,linecolor=niceblue!15,%
		linewidth=1pt,topline=true,%
		frametitleaboveskip=\dimexpr-\ht\strutbox\relax%
	}
	\begin{mdframed}[]\relax}{%
\end{mdframed}}

\newcommand{\keepvalues}{%
	\edef\restorevalues{%
		\parindent=\the\parindent
		\parskip=\the\parskip
	}%
}

\newtheorem{theorem}{Theorem}[section]
\newtheorem{corollary}{Corollary}[theorem]
\newtheorem{lemma}[theorem]{Lemma}
\newtheorem{proposition}[theorem]{Proposition}
\newtheorem{notation}[theorem]{Notation}
\newtheorem{definition}[theorem]{Definition}
\newtheorem{fact}[theorem]{Fact}

\newcounter{theo}[section]
\renewcommand{\thetheo}{\arabic{theo}}

\newcounter{openq}[section]
\renewcommand{\theopenq}{\arabic{openq}}

\newcommand{\defsection}{\section}
\newcommand{\defsubsection}{\subsection}
\newcommand{\defsubsubsection}{\subsubsection}

\newcommand{\subsectionname}{Subsection}
\newcommand{\subsubsectionname}{Subsubsection}
\newcommand{\sectionnamelower}{section}
\newcommand{\subsectionnamelower}{subsection}

\title{Bipartite Perfect Matching as a Real Polynomial}

\author{Gal Beniamini
			\thanks{School of Computer Science and Engineering, The Hebrew University of Jerusalem}
			\footnote{This project has received funding from the European Research Council (ERC) under the European Union's Horizon 2020 research and innovation programme (grant agreement No 740282)}
		\quad
		Noam Nisan $^{*}$$^\dagger$}

\begin{document}

\maketitle

\begin{abstract}
	We obtain a description of the Bipartite Perfect Matching decision problem as a multilinear polynomial over the Reals. We show that it has full degree and $(1-o_n(1))\cdot 2^{n^2}$ monomials with non-zero coefficients.  
	In contrast, we show that in the dual representation (switching the roles of 0 and 1) the number of monomials is \textit{only} exponential in $\Theta(n \log n)$. Our proof relies heavily on the fact that the lattice of graphs which are ``matching-covered'' is Eulerian.
\end{abstract}
\defsection{Introduction}

Every Boolean function $f:\{0,1\}^n \rightarrow \{0,1\}$ can be represented in a unique way as a Real multilinear polynomial. This representation and related ones (e.g. using the $\{1,-1\}$ basis rather than $\{0,1\}$ -- the ``Fourier transform'' over the hypercube, or approximation variants) have many applications for various complexity and algorithmic purposes. See, e.g., \cite{o2014analysis} for a recent textbook.

In this paper we derive the representation of the bipartite-perfect-matching decision problem as a Real polynomial.

\begin{thinleftbar}
	\begin{definition*}
		The Boolean function $BPM_n(x_{1,1}, \dots, x_{n,n})$ is defined to be $1$ if and only if the bipartite graph whose edges are 
		$\set{(i,j)}{x_{i,j}=1}$ has a perfect matching, and $0$ otherwise.
	\end{definition*}
\end{thinleftbar}

Our first result is determining the representation of this function as a Real multilinear polynomial. By way of example, $BPM_2(\bar{x})=x_{1,1}x_{2,2}+x_{1,2}x_{2,1}-x_{1,1}x_{1,2}x_{2,1}x_{2,2}$.
Somewhat surprisingly, finding the closed form expression for any $n$ appears nontrivial. In fact,
we do not know of an easier proof than our own involved proof, even showing that for any 
$n$ the degree of this polynomial is $n^2$. \footnote{For the special case where $n$ is a prime power, 
the full degree of the polynomial
follows from an extension of the evasiveness result of \cite{rivest1975generalization}, due to \cite{NSS2008}.  However, for $n$ that is not a prime power, it is not true that every
monotone bipartite graph property has a full degree.}  

To present our first result, let us introduce some notation. We will call a graph \emph{matching-covered} if its edges can be represented as a union of perfect matchings. As an example, for $n=2$ the graph whose edges are $\{(1,1),(1,2),(2,2)\}$ is
{\em not} matching-covered since any perfect matching that contains the edge $(1,2)$ must also contain the edge $(2,1)$, which is not in the graph. The connected components of
matching-covered graphs are called ``elementary graphs'' and were studied at length by \cite{plummer1986matching}.
Finally for a graph $G$, we denote its cyclomatic number by $\chi(G)=|E(G)|-|V(G)|+|C(G)|$ where $|C(G)|$ is the number of connected components of $G$. The following Theorem characterizes the multilinear polynomial of $BPM_n$.

\begin{theo}[The Bipartite Perfect Matching Polynomial]{thm:bpm_poly}
	\label{bpm_poly}
	\begin{equation*}BPM_{n}(x_{1,1}, \dots, x_{n,n}) = \sum_{G \subseteq K_{n,n}} a_G \prod_{(i,j) \in E(G)} x_{i,j}\text{ ,  where:}\end{equation*}
	\[a_G = \begin{cases}
				0 &\mbox{if } G \text{ is not matching-covered} \\
				(-1)^{\chi(G)} & \mbox{if } G \text{ is matching-covered}
			\end{cases}\] 
\end{theo}
	
Our proof proceeds by studying the structure of the lattice of 
matching-covered graphs and its M\"obius function and
the key step requires using the topological structure of this lattice. Specifically, \cite{billera1994combinatorics} showed that this lattice is isomorphic to the face lattice of the Birkhoff Polytope, and is thus Eulerian. Counting the number of matching-covered graphs, we get:
	
\begin{thinleftbar}
\begin{corollary*}
	The polynomial $BPM_n$ has $(1-o_n(1)) \cdot 2^{n^2}$ monomials with non-zero coefficients.
\end{corollary*}
\end{thinleftbar}

Our characterization of the polynomial has several corollaries. For example, it allows us to obtain a closed form expression counting the number of bipartite graphs containing a perfect matching, and in particular to show that this number is odd. It also suffices for showing that a $(1 - o_n(1))$-fraction of the Fourier coefficients of $BPM_n$ are very small, $2^{-n^2 + 1}$, yet non-zero.

In the second part of the paper, we turn our attention towards the ``dual representation'' -- a form in which the symbols $1$ and $0$ switch roles.
Formally, for a Boolean function $f(x_1, \dots, x_n)$ we define its dual by $f^{\star}(x_1, \dots, x_n) = 1-f(1-x_1, \dots, 1-x_n)$. Under this notation, $BPM^{\star}_n(x_{1,1}, \dots, x_{n,n})$ gets the value $1$
if the input graph contains a biclique over a total of $n+1$ vertices (i.e., its complement contains
a violation of Hall's condition).
	
To present our result, we will focus on the following two classes of graphs. A bipartite graph is called {\em totally ordered} if there exists an ordering $v_1, \dots, v_n$ of its left vertices such that
$N(v_1) \supseteq N(v_2) \supseteq \dots \supseteq N(v_n)$ where $N(v)$ denotes the set of right
vertices connected to $v$.  In the same vein, we call the graph
{\em strictly totally ordered} if in fact $N(v_1) \supsetneq N(v_2) \supsetneq \dots \supsetneq N(v_n) \supsetneq \emptyset$.
For the dual case, we do not obtain a complete characterization of the polynomial. Nevertheless, we show the following fine grained characterization.

\begin{theo}[The Dual Polynomial of Bipartite Perfect Matching]{thm:dual_bpm_poly}
	\label{dual_bpm_poly}
	\begin{equation*}BPM^{\star}_{n}(x_{1,1}, \dots, x_{n,n}) = \sum_{G \subseteq K_{n,n}} a^{\star}_G \prod_{(i,j) \in E(G)} x_{i,j}\text{ ,  where:}\end{equation*}
	\begin{itemize}
		\item If $G$ \textit{is not totally ordered}, we have $a^\star_G = 0$.
		\item If $G$ \textit{is strictly totally ordered}, we have $a^\star_G = (-1)^{n+1}$
	\end{itemize}
\end{theo}

Our proof relies on properties of the lattice of matching-covered graphs, and heavily utilizes its Eulerian structure. For graphs $G$ that are totally ordered but not strictly so, the situation is complex. We show
that for some such graphs $G$, we have $a^\star_G=0$, for others $a^\star_G=\pm 1$, and for others still $a^\star_G \not\in \{-1,0,1\}$. For example, for $n > 2$ and $G=K_{n-1,n-1}$ we have $a^\star_G=(n-2)^2$. We present the full polynomial of $BPM^\star_3$ in Appendix \ref{bpm_star_3}. We leave the full characterization of the dual polynomial as an open problem.

This characterization of the dual polynomial suffices for
obtaining an accurate estimate of the number of monomials with non-zero coefficients:

\begin{thinleftbar}
	\begin{corollary*}
		The polynomial $BPM^*_n$ has $2^{2n \cdot \log_2(n) + \mathcal{O}(n)}$ monomials with non-zero coefficients.
	\end{corollary*}
\end{thinleftbar}

We view the small number of non-zero coefficients as some form of a positive algorithmic result
regarding the perfect matching problem. For example, consider a communication setting where
the edges of a bipartite graph are partitioned somehow between two parties; Alice and Bob. Their task is to devise a communication protocol for determining whether the combined graph has a perfect matching.
The known algorithms for bipartite matching imply a protocol that uses $\mathcal{O}(n^{1.5})$ bits
of communication \cite{dobzinski2019economic,nisan2019demand}. However, the small number of monomials in $BPM^\star_n$ directly implies that
the associated communication matrix has Real rank that is \textit{only} exponential in $n \log n$ (recall that the logarithm of the rank is a lower bound for the deterministic communication complexity, and is conjectured to be polynomially related to it).

Conversely, the polynomial representations of $BPM_n$ and $BPM_n^\star$ allow us to obtain \textit{new lower bounds} on the decision problem of bipartite perfect matching. In particular, we consider three families of decision trees; those whose internal nodes are labeled by $XOR$, $AND$ or $OR$ functions, respectively. Of particular note is the family of $OR$ decision trees, which were shown by \cite{nisan2019demand} 
to be complexity preserving proxies for many efficient algorithms for bipartite perfect matching. The known algorithms for the problem imply that $\tilde{\mathcal{O}}(n^{1.5})$ $OR$ queries suffice (even when slightly restricting each query), and any $\Omega(n^{1 + \alpha})$ lower bound would rule out asymptotically fast algorithms from a wide class, i.e., ``combinatorial algorithms''.

To present our lower bounds, we introduce the following notation. For each of the three families outlined above, we denote the minimal depth of a tree in the family computing $BPM_n$ by $D^{XOR}(BPM_n)$, $D^{AND}(BPM_n)$ and $D^{OR}(BPM_n)$, respectively.

\begin{thinleftbar}
	\begin{corollary*} 
	$BPM_n$ is \textit{evasive} for $XOR$ decision trees, i.e., $D^{XOR}(BPM_n) = n^2$.$\quad\quad\quad\quad$ \newline Furthermore, for $AND$ and $OR$ decision trees, we have the following lower bounds: 
	\begin{equation*}
		\begin{split}
			D^{AND}(BPM_n) &\ge (\log_3 2) \cdot n^2 - o_n(1)
		\end{split}
		\quad\quad\quad
		\begin{split}
			D^{OR}(BPM_n) &\ge 2 \log_3(n!)	
		\end{split}
	\end{equation*}
	\end{corollary*}
\end{thinleftbar}

\newpage
\newgeometry{bottom=4cm} {
\tableofcontents
}
\restoregeometry
\newpage
\defsection{Preliminaries and Notation}

\defsubsection{Polynomial Representations of Boolean Functions}

Recall the following fact regarding polynomial representations of Boolean functions (see \cite{o2014analysis}):

\begin{leftbar}
\begin{fact}
	\label{uniqueboolean}
	Any Boolean function $f: \{0,1\}^n \rightarrow \{0,1\}$ can be \textit{uniquely} represented by a multilinear polynomial over the Reals.
\end{fact}
\end{leftbar}

\noindent For a given a multilinear polynomial, we denote the set of all monomials appearing in it by:
\begin{leftbar}
	\begin{notation}
		\label{mon_notation}
		Let $f(x_1, \dots, x_n) = \sum_{S \subseteq [n]} a_S \left(\prod_{i \in S} x_i\right) \in \mathbb{R}[x_1, \dots, x_n]$ be a multilinear polynomial over the Reals. Denote the set of monomials appearing in $f$ by:
		\[ mon(f) = \set{S \subseteq [n]}{a_S \ne 0} \] 
	\end{notation}
\end{leftbar}

\defsubsection{The M\"obius Function of Partially Ordered Sets}

%

When discussing partially ordered sets (hereafter, \textbf{posets}), we use the M\"obius function for posets. The M\"obius function of a poset is the inverse, with respect to convolution, to the poset's zeta function $\zeta(y,x) = \mathbbm{1}\{y < x\}$ (see, e.g., \cite{Stanley:2011:ECV:2124415}).

\begin{leftbar}
\begin{definition}[M\"obius Function for Posets]
	\label{mobius_func_poset}
	Let $\mathcal{P} = (P, <)$ be a finite poset. The M\"obius function of the poset $\mathcal{P}$ is denoted by $\mu_P : P \times P \rightarrow \mathbb{R}$, and is defined as follows:
	\begin{equation*}
	\begin{split}
	\forall x \in P:\ \mu_P(x, x) &= 1\\
	\forall x,y \in P,\ y < x:\ \mu_P(y,x) &= - \sum_{y \le z < x} \mu_P(y,z)
	\end{split}
	\end{equation*}
\end{definition}
\end{leftbar}

Given a poset $\mathcal{P}$ with a unique bottom element $\hat{0}$, the values $\mu_P(\hat{0}, x)$, where $x \in \mathcal{P}$, are known as the \textbf{M\"obius Numbers} of $\mathcal{P}$.

%

\defsubsection{Graphs}

We use the standard definitions and notation relating to graphs. For a graph $G$, we denote the sets of vertices and edges of $G$ by $V(G)$ and $E(G)$, respectively. The set of all perfect matchings of $G$ is denoted by $PM(G)$, and the set of all connected components is denoted by $C(G)$. Furthermore, for any vertex $v \in V(G)$, we denote its neighbour set by $N_G(v)$. 

In addition to the quantities relating to a given graph, it will be useful to also provide some notation for basic operations on graphs. For example, the notations $G \cup \{(a,b)\}$ and $G \setminus \{(a,b)\}$ refer to the graph $G$ with the addition or removal of the edge $(a,b)$, respectively. In the same vein, $G - a$ is the graph where the vertex $a$ is omitted, along with all the edges adjacent to it. Lastly, if $H$ and $G$ are two graphs, the notation $H \subseteq G$ indicates that $E(H) \subseteq E(G)$ and $V(H) = V(G)$.

A somewhat less common quantity which we refer to throughout the paper is the \textbf{Cyclomatic Number} of the graph, which is defined as follows:
 
\begin{leftbar}
	\begin{definition}
		\label{cyclomatic_number}
		Let $G$ be a graph. The cyclomatic number of $G$, $\chi(G)$, is defined:
		
		\[ \chi(G) = |E(G)| - |V(G)| + |C(G)| \]
	\end{definition}
\end{leftbar}

We will often consider the edge sets corresponding to unions of graphs. Consequently, the following notation will be useful:

\begin{leftbar}
	\begin{notation}
		Let $S$ be a set of graphs. The set of all edges appearing in any graph $G \in S$ is denoted by: 
			\[\bar{E}(S) = \bigcup_{G \in S} E(G)\]
	\end{notation}
\end{leftbar}

Lastly, when dealing with Boolean graph functions (i.e., Boolean functions whose input bits correspond to the edges of graphs over a fixed set of vertices), we use the following notation:

\begin{leftbar}
	\begin{notation}
		Let $n,m \in \mathbb{N}^+$. Let $f: \{0,1\}^{nm} \rightarrow \{0,1\}$ be a Boolean function whose inputs are bipartite graphs over the vertices of $K_{n,m}$. Then, $\forall G \subseteq K_{n,m}$ denote:
		\[ f(G) \eqdef f(x_G),\text{ where } \forall i \in [n],\ \forall j \in [m]:\ (x_G)_{i,j} = \mathbbm{1}\{(i,j) \in E(G)\} \]
	\end{notation}
\end{leftbar}

\defsubsection{Decision Trees and Query Complexity}
\label{decision_trees}

Decision trees are binary trees whose internal nodes are labeled by Boolean functions, and whose leaves are labeled by the values $\{0,1\}$. Formally, we say that a decision tree $T$ \textbf{computes} a Boolean function $f: \{0,1\}^n \rightarrow \{0,1\}$ if for any root-to-leaf path in $T$, the value of the leaf ``agrees'' with $f(z)$ on all inputs $z \in \{0,1\}^n$ which are \textit{processed} by the path. An input $z \in \{0,1\}^n$ is \textit{processed} by a path if for all functions $h$ in the internal nodes along the path, we have $h(z) = 1$ if the path turns right at that node, and $h(z) = 0$ otherwise.

From an algorithmic perspective, decision trees can be viewed as algorithms whose every step consists of \textit{querying} the output of some Boolean function $h \in \mathcal{H}$ on the input bits, and repeating the process until sufficient \textit{information} is available to deduce the output. Thus, decision trees give rise to the \textit{query complexity model}. In this model, we disregard the amount of computation required, and instead measure the minimal amount of \textit{information}. There are several families of decision trees, which differ from one another in the set of functions $\mathcal{H}$ which label their internal nodes.

\paragraph{Classical Decision Trees} These are decision trees whose internal nodes are labeled by dictatorship functions, i.e., each internal node ``queries'' the value of a single input bit. For such trees, we use the following measure of complexity:

\begin{leftbar}
\begin{definition}
	 Let $f: \{0,1\}^n \rightarrow \{0,1\}$ be a Boolean function. The \textit{minimal depth} of a classical decision tree computing $f$ is known as the \textbf{Query Complexity of f}.
\end{definition}
\end{leftbar}

\paragraph{Generalized Decision Trees} Three natural extensions of classical decision trees are those whose internal nodes are labeled by $XOR$, $OR$ and $AND$ functions, respectively, over arbitrary subsets of the input bits. $XOR$ decision trees have been studied at length, and are known to be related to the Fourier expansion of a function. $OR$ and $AND$ decision trees have also been studied, for example in the setting of group property testing. For these three families of trees, we denote their corresponding query complexities as follows:

\begin{leftbar}
	\begin{definition}
		Let $f: \{0,1\}^n \rightarrow \{0,1\}$ be a Boolean function. We denote the \textbf{minimal depth} of any $XOR$, $OR$ or $AND$ decision tree computing $f$, by $D^{XOR}(f)$, $D^{OR}(f)$ and $D^{AND}(f)$, respectively. 
	\end{definition}
\end{leftbar}

\defsubsection{Fourier Analysis}
\label{decision_trees}

Fourier Analysis of Boolean functions is a wide field of study, in which powerful analysis tools are applied to functions over the Hamming cube, yielding combinatorial (and other) insights. Given a Boolean function $f: \{-1,1\}^n \rightarrow \{-1,1\}$, the Fourier expansion of $f$ is the unique multilinear polynomial representing $f$ over the Reals in the $\{1, -1\}$ basis (i.e., $-1$ corresponds to $True$ and $1$ to $False$). The Fourier expansion of $f$ is given by:

\[ f(x_1, \dots, x_n) = \sum_{S \subseteq [n]} \hat{f}_S \cdot \prod_{i \in S} x_i \]

Where each $\hat{f}_S$ is a Real number, referred to as the Fourier coefficient of $S$, and each monomial $\prod_{i \in S} x_i$ corresponds to a parity function over the set $S$. The aforementioned representation is unique, and the set of Fourier coefficients of $f$ is commonly referred to as its \textit{Fourier Spectrum}. Crucially, the set of all monomials forms an orthonormal basis. There are many important properties of the Fourier expansion, which we will not recount here. For an extensive treatment of the topic, we refer the reader to \cite{o2014analysis}.

\defsection{The Boolean Bipartite Perfect Matching Polynomial}
\label{primal_section}

This {\sectionnamelower} centers around the proof of Theorem \ref{bpm_poly}. We begin with some basic observations regarding a family of Boolean graph functions called ``Graph Cover functions''. These observations lead us to the connection between the multilinear polynomial representing $BPM_n$, and the M\"obius numbers of the lattice of matching-covered graphs. To compute these M\"obius numbers, we rely on a result of Billera and Sarangarajan \cite{billera1994combinatorics}, showing that the aforementioned lattice is isomorphic to the face lattice of the Birkhoff Polytope. 

Using Theorem \ref{bpm_poly}, we deduce several corollaries. For example, we find a closed form expression counting the number of bipartite graphs having a bipartite perfect matching, and deduce that this number is odd. We also compute asymptotically almost all the Fourier spectrum of $BPM_n$. Lastly, we obtain new lower bounds for decision trees; we show that $BPM_n$ is ``evasive'' for $XOR$ decision trees (i.e., exactly $n^2$ queries are required), and that for $AND$ decision trees, at least $(\log_3 2) \cdot n^2 - o_n(1)$ queries are required. 

\defsubsection{Graph Cover Functions}
\label{graph_cover_functions}

Let $\mathcal{H}$ be a set of labeled graphs over a fixed common vertex set. Consider the following natural Boolean graph function: ``Given a labeled graph $G$ over the same vertex set, does $G$ contain any graph in $\mathcal{H}$ as a subgraph?''. In what follows, we restrict our discussion to bipartite graphs and fix our vertex set to be the vertices of the complete bipartite graph, $K_{n,m}$. Nevertheless, the same observations apply to general graphs. Formally, we define the Graph Cover function of $\mathcal{H}$ as follows:

\begin{leftbar}
\begin{definition}
	\label{graph_cover_func}
	Let $\mathcal{H}$ be a set of bipartite graphs over the vertices of $K_{n,m}$. The \textbf{Graph Cover} function of $\mathcal{H}$, $f_{\mathcal{H}}:\ \{0,1\}^{nm} \rightarrow \{0,1\}$, is defined as follows:
	
	\[ \forall G \subseteq K_{n,m}:\ f_{\mathcal{H}}(G) = \mathbbm{1}\{\exists H \in \mathcal{H},\ H \subseteq G\} \]
	
\end{definition}
\end{leftbar}

Given a set of graphs $\mathcal{H}$ over the vertices of $K_{n,m}$ and a graph $G \subseteq K_{n,m}$, we say that $G$ is \textbf{$\mathcal{H}$-covered} if there exists some $\emptyset \ne S \subseteq \mathcal{H}$ such that $\bar{E}(S) = E(G)$. Moreover, we denote by $\mathcal{C}(\mathcal{H})$ the set of \textit{all} $\mathcal{H}$-covered graphs. The following simple observation regarding the \textit{monomials} of the multilinear polynomial representing $f_\mathcal{H}$ can be made.

\begin{leftbar}
	\begin{proposition}
		\label{monomials_graph_cover_functions}
		Let $\mathcal{H}$ be a set of bipartite graphs over the vertices of $K_{n,m}$. The only monomials appearing in the multilinear polynomial representing $f_\mathcal{H}$ over the Reals are those corresponding to $\mathcal{H}$-covered graphs.
	\end{proposition}
\end{leftbar}
\begin{proof}
	The DNF formula representing the graph cover function is:
	
	\[\varphi = \bigvee_{H \in \mathcal{H}}\ \bigwedge_{(i,j) \in E(H)}\ x_{i, j}\]
	
	Since each $x_{i,j} \in \{0,1\}$, we have $\forall k \in \mathbb{N}^+$, $x_{i,j}^k = x_{i,j}$. Therefore, arithmetizing the formula yields the following polynomial representation:
	
	\begin{equation*}
		\begin{split}
			f_{\mathcal{H}}(x_{1,1}, \dots, x_{n,m}) &= 1 - \underset{H \in \mathcal{H}}{\prod}(1 - \underset{(i,j) \in E(H)}{\prod} x_{i, j}) \\
			&= \sum_{\emptyset \ne S \subseteq \mathcal{H}} (-1)^{|S|+1} \prod_{G \in S} \ \prod_{(i,j) \in E(G)} x_{i, j} \\
			&= \sum_{G \in \mathcal{C}(\mathcal{H})} \left(\sum_{\substack{\emptyset \ne S \subseteq \mathcal{H}\\ \bar{E}(S)=E(G)}}(-1)^{|S|+1}\right) \prod_{(i,j) \in E(G)} x_{i, j} \qedhere
		\end{split}
	\end{equation*}
\end{proof}

The set of $\mathcal{H}$-covered graphs, together with the subset relation over edges, form a partially ordered set. This partially ordered set has two important properties. Firstly, it is a lattice; every two elements have a unique supremum (``join'') and a unique infimum (``meet''). Secondly, the M\"obius numbers of this lattice exactly describe the \textit{coefficients} of the multilinear polynomial representing the graph cover function, $f_\mathcal{H}$. 

\begin{leftbar}
	\begin{proposition}
		\label{graph_cover_lattice}
		Let $\mathcal{H}$ be a set of bipartite graphs over a fixed vertex set. The poset $\mathcal{P} = (\mathcal{C}(\mathcal{H}) \cupdot \{\hat{0}\}, \subseteq)$ is a bounded lattice, where $\hat{0}$ is the empty graph. 
	\end{proposition}
\end{leftbar}
\begin{proof}
	The subset relation over the edges is reflexive, transitive and anti-symmetric, thus $\mathcal{P}$ is a poset. Furthermore, $\mathcal{P}$ is bounded, since $\hat{0} = (V(\mathcal{H}), \emptyset)$ and $\hat{1} = (V(\mathcal{H}), \bar{E}(\mathcal{H}))$. It remains to show that $\forall G_1,G_2 \in \mathcal{C}(\mathcal{H})$ there exists a \textit{join} (unique supremum) and a \textit{meet} (unique infimum).
	
	Let $G_1,G_2 \in \mathcal{C}(\mathcal{H})$. The meet and join of $G_1$ and $G_2$ are given by:
	
	\begin{equation*}
	\begin{split}
	E(G_1 \vee G_2) &= \bigcup_{\substack{H \in \mathcal{H}\\(H \subseteq G_1) \lor (H \subseteq G_2)}} E(H) = E(G_1) \cup E(G_2) \\
	E(G_1 \wedge G_2) &= \bigcup_{\substack{H \in \mathcal{H}\\(H \subseteq G_1) \land (H \subseteq G_2)}} E(H)
	\end{split}
	\end{equation*}
	
	For the join operator, let $G \eqdef G_1 \vee G_2$. By construction, $G_1 \subseteq G$ and $G_2 \subseteq G$, therefore $G$ is a supremum. Assume towards a contradiction that there exists another supremum $\hat{G} \ne G$ such that $G \not\subseteq \hat{G}$. Let $x \in E(G) \setminus E(\hat{G})$. Without loss of generality, assume $x \in E(G_1)$. Then $x \in E(G_1)$ and $x \notin E(\hat{G})$ therefore $G_1 \not\subseteq \hat{G}$, in contradiction to the fact that $\hat{G}$ is a supremum.
	
	For the meet operator, let $G \eqdef G_1 \wedge G_2$. By construction, $G \subseteq G_1$ and $G \subseteq G_2$, therefore $G$ is an infimum. Assume towards a contradiction that there exists another infimum $\hat{G} \ne G$ such that $G \not\supseteq \hat{G}$. Let $x \in E(\hat{G}) \setminus E(G)$. Since $\hat{G} \in \mathcal{C}(\mathcal{H})$, there exists $H_x \in \mathcal{H}$ such that $H_x \subseteq \hat{G}$, $x \in E(H_x)$. However, $\hat{G}$ is an infimum, thus $H_x \subseteq \hat{G} \subseteq G_1$ and $H_x \subseteq \hat{G} \subseteq G_2$, thus by construction $H_x \subseteq G$ and $x \in E(G)$, a contradiction. 
\end{proof}

\begin{leftbar}
\begin{proposition}
	\label{graph_cover_polynomial_mobius}
	Let $\mathcal{H}$ be a set of bipartite graphs over the vertices of $K_{n,m}$ and let $\mathcal{P} = (\mathcal{C}(\mathcal{H})\cupdot \{\hat{0}\}, \subseteq)$ be the graph cover lattice of $\mathcal{H}$. Then:
	
	\[ f_{\mathcal{H}}(x_{1,1}, \dots, x_{n,m}) = \sum_{G \in \mathcal{C}(\mathcal{H})} {-\mu_P(\hat{0}, G) \cdot \prod_{(i,j) \in E(G)} x_{i, j}} \]
	
	\noindent Namely, the coefficients of the multilinear polynomial representing $f_{\mathcal{H}}$ over the Reals are given by the (negated) M\"obius numbers of $\mathcal{P}$.
\end{proposition}
\end{leftbar}
\begin{proof}
	Let $f$ be the polynomial $f(x_{1,1}, \dots, x_{n,m}) = \sum_{G \in \mathcal{C}(\mathcal{H})} {-\mu_P(\hat{0}, G) \cdot \prod_{(i,j) \in E(G)} x_{i, j}}$, and let $H \subseteq K_{n,m}$ be a graph. Denote by $H'$ the union of all graphs $G \in \mathcal{C}(\mathcal{H})$ such that $G \subseteq H$. We now show that $f$ agrees with $f_{\mathcal{H}}$ on all inputs, and deduce the identity by the uniqueness of the representing polynomial. If $H' = \hat{0}$, then indeed $f(H) = 0$ as required. Otherwise, we have:
	
	\begin{equation*}
		\begin{split}
			f(H) &= \sum_{G \in \mathcal{C}(\mathcal{H})} {-\mu_P(\hat{0}, G) \cdot \mathbbm{1}\{G \subseteq H\}} \\
				&= \sum_{\substack{
						\hat{0} \subset G \subseteq H' \\ 
						G \in \mathcal{C}(\mathcal{H})}} -\mu_P(\hat{0}, G) \\
				&= \sum_{\substack{
						\hat{0} \subseteq G \subseteq H' \\ 
						G \in \mathcal{C}(\mathcal{H})}} -\mu_P(\hat{0}, G) + \mu_P(\hat{0}, \hat{0})
		\end{split}
	\end{equation*}
	
	\noindent And by the definition of the M\"obius function, $\mu_P(\hat{0}, \hat{0}) = 1$ and $\displaystyle\sum_{\substack{
			\hat{0} \subseteq G \subseteq H' \\ 
			G \in \mathcal{C}(\mathcal{H})}} -\mu_P(\hat{0}, G) = 0$
\end{proof}

\defsubsection{The Boolean Bipartite Perfect Matching Polynomial}
\label{perfect_matching_polynomial_subsection}

\defsubsubsection{Matching-Covered and Elementary Graphs}

Let us begin by recalling the definition of the Boolean Bipartite Perfect Matching function:

\begin{leftbar}
	\begin{definition*}
		The Boolean Bipartite Perfect Matching function, $BPM_n: \{0,1\}^{n^2} \rightarrow \{0,1\}$, is defined as follows: \[BPM_n(x_{1,1}, \dots, x_{n,n}) = \begin{cases}
		 	1 & \set{(i,j)}{x_{i,j}=1} \text{ has a Perfect Matching} \\
		 	0 & Otherwise
		 \end{cases}\]
	\end{definition*}
\end{leftbar}

The monotone Boolean function $BPM_n$ represents the \textbf{decision problem} of bipartite perfect matching. Given a bipartite graph $G \subseteq K_{n,n}$, the function outputs $1$ if and only if $G$ contains a bipartite perfect matching. The aforementioned function may also be cast in terms of graph cover functions. In particular, it is a graph cover function for the set $\mathcal{H} = PM(K_{n,n})$. Thus, by Proposition \ref{monomials_graph_cover_functions}, the only monomials that may appear in its multilinear polynomial over the Reals are those corresponding to $\mathcal{H}$-covered graphs. For the particular case where $\mathcal{H} = PM(K_{n,n})$, we introduce the following definition:

\begin{leftbar}
\begin{definition}
	\label{matching_covered}
	Let $G \subseteq K_{n,n}$ be a balanced bipartite graph. $G$ is \textbf{matching-covered} if and only if there exists some $S \subseteq PM(K_{n,n})$ such that $\bar{E}(S) = E(G)$. 
\end{definition}
\end{leftbar}

For simplicity, we introduce some notation. The set of all matching-covered graphs $H \subseteq G$ is denoted by $\mathbf{MC(G)}$. In the same vein, the set of \textit{all} bipartite matching-covered graphs of order $2n$ is denoted $\mathbf{MC_n} \eqdef MC(K_{n,n})$. Lov{\'a}sz and Plummer \cite{plummer1986matching} previously considered a family of graphs called \textit{elementary graphs}, which are closely related to matching-covered graphs. Elementary graphs are simply the \textit{connected components} of matching-covered graphs. Formally:

\begin{leftbar}
\begin{definition}[\cite{plummer1986matching}]
	\label{elementary_graph}
	$G$ is \textbf{elementary} $\Leftrightarrow$ $G$ is a \textit{connected} matching-covered graph. 
\end{definition}
\end{leftbar}

\begin{figure}[h]
	\centering
	\includegraphics[height=4cm]{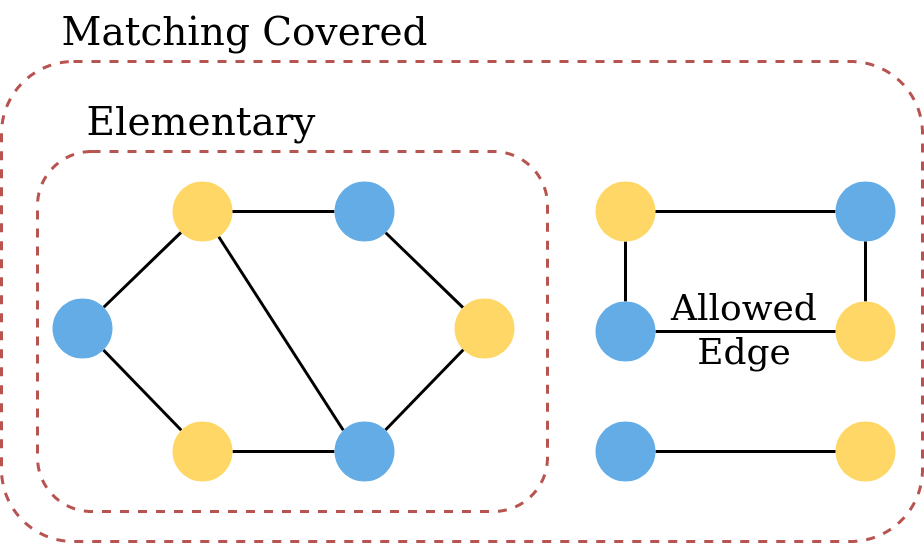}
	\caption{A matching-covered graph, composed of three elementary graphs}
\end{figure}

We recall two key Theorems regarding elementary graphs. The first, due to  Hetyei \cite{hetyei1964rectangular}, provides several necessary and sufficient conditions for elementarity of a given graph. The second, due to Lov{\'a}sz and Plummer \cite{plummer1986matching}, shows that all elementary graphs admit a normal form, called the \textit{bipartite ear decomposition}.

\begin{leftbar}
	\begin{theorem}[\cite{hetyei1964rectangular}]
		\label{hetyei_conditions}
		Let $G=(A \cupdot B, E)$ be a bipartite graph. The following are equivalent:
		\begin{itemize}
			\item $G$ is elementary.
			\item $G$ has exactly two minimum vertex covers, $A$ and $B$.
			\item $|A|=|B|$ and for every $\emptyset \ne X \subset A$, $|N(X)| \ge |X|+1$.
			\item $G=K_2$, or $|V(G)| \ge 4$ and for any $a \in A$, $b \in B$, $G - a - b$ has a perfect matching.
			\item $G$ is connected and every edge is ``\textit{allowed}'', i.e., appears in a perfect matching of $G$.
		\end{itemize}
	\end{theorem}
\end{leftbar}

\begin{leftbar}
\begin{definition}[\cite{plummer1986matching}]
	\label{bipartite_ear_decomp}
	Let $G$ be a balanced bipartite graph. $G$ has a \textbf{bipartite ear decomposition} of length $k$ if it can be written in the form:
	\[ G = e + P_1 + \dots + P_k \]
	
	\noindent Where $e \in E(G)$, and each $P_i$ is an odd-length path, in which any pair of adjacent vertices are from different colour classes (and in particular, so are its endpoints). The vertices appearing in each path $P_i$, other than its two endpoints, are ``\textit{fresh}'' -- i.e., they do not appear in $e + P_1 + \dots + P_{i-1}$. Note that each $P_i$ can also be a single edge connecting two preexisting vertices of different colour classes. 
\end{definition}
\end{leftbar}

\begin{leftbar}
\begin{theorem}[\cite{plummer1986matching}]
	\label{elementary_iff_ear_decomp}
	Let $G$ be a balanced bipartite graph. Then:
	\[ G \text{ is elementary} \iff G \text{ has a bipartite ear decomposition} \]
\end{theorem}
\end{leftbar}

Using a probabilistic method argument, we observe that the vast majority of balanced bipartite graphs are in fact elementary (and in particular, matching-covered):

\begin{leftbar}
	\begin{proposition}
		\label{number_of_mu_graphs}
		Let $n > 1$. Then:
		
		\[|MC_n| \ge |\set{G \subseteq K_{n,n}}{G \text{ is elementary}}| \ge 2^{n^2} \left(1 - \frac{2 n^4}{2^n}\right) = 2^{n^2}(1 - o_n(1))\]
	\end{proposition}
\end{leftbar}
\begin{proof}
	Let $n > 1$ and let $A,B$ be two sets, where $|A|=|B|=n$. Denote by $G(n,n,p)$ the distribution over balanced bipartite graphs of order $2n$, in which each edge appears i.i.d with probability $p$. Recall that by Theorem \ref{hetyei_conditions}, $G=(A \cupdot B, E) \subseteq K_{n,n}$ is elementary if and only if $\forall a \in A$, $\forall b \in B$: $G - a - b$ has a perfect matching. By the union bound:
	
	\begin{equation*}
	\begin{split}
	\Pr_{G \sim G(n,n,0.5)}\left[\text{G is not elementary}\right] &= \Pr_{G \sim G(n,n,0.5)}\left[\exists a \in A,\ b \in B:\ G - a - b\text{ has no perfect matching}\right] \\
	&\le n^2 \cdot \Pr_{G \sim G(n-1,n-1,0.5)}\left[\text{G has no perfect matching}\right]
	\end{split}
	\end{equation*}
	
	By Hall's Theorem, $G$ has a perfect matching if and only if $\forall X \subseteq A$: $|N(X)| \ge |X|$. Thus $G$ has no perfect matching if and only if there exist two sets $S \subseteq A$, $T \subseteq B$ such that $|S| + |T| = n+1$, and none of the edges in $S \times T$ appear in $G$. Using the union bound again:
	
	\begin{equation*}
	\begin{split}
	\Pr_{G \sim G(n,n,0.5)}\left[\text{G has no perfect matching}\right] \le \sum_{k=1}^{n} {n \choose k} {n \choose k-1} 2^{-k(n-k+1)} \le \frac{n^2}{2^n}
	\end{split}
	\end{equation*}
	Thus:
	\begin{equation*}
	\begin{split}
	\Pr_{G \sim G(n,n,0.5)}\left[G \in MC_n\right] \ge \Pr_{G \sim G(n,n,0.5)}\left[\text{G is elementary}\right] \ge 1 - \frac{2n^4}{2^n} \qedhere
	\end{split}
	\end{equation*}
\end{proof}

\defsubsubsection{The Birkhoff Polytope and the Matching-Covered Lattice}

Let $P$ be a polytope. The \textit{face lattice} of $P$ is the lattice whose elements are the faces of $P$, ordered by containment, together with a unique bottom element $\hat{0}$ (i.e., the ``empty face'') and a unique top element $\hat{1}$ (corresponding to the polytope $P$ itself). The aforementioned lattice is \textit{ranked}, and the rank of each face $Q \ne \hat{0}$ is given by $dim(Q) + 1$.

We now recall a particular polytope: the \textbf{Birkhoff polytope}, $B_n$. This polytope is defined as the convex hull of all $n \times n$ permutation matrices. Billera and Sarangarajan proved the following powerful theorem regarding the face lattice of $B_n$:

\begin{leftbar}
	\begin{theorem}[\cite{billera1994combinatorics}]
		\label{elementary_faces_of_birkhoff}
		The face lattice of the Birkhoff polytope $B_n$ is isomorphic to the lattice of all matching-covered graphs of order $2n$, ordered by inclusion, together with the empty graph.
	\end{theorem}
\end{leftbar}


A lattice that is isomorphic to the face lattice of a polytope is known as ``Eulerian''. The M\"obius function of an Eulerian lattice satisfies the following identity (see, e.g., \cite{Stanley:2011:ECV:2124415}): $\forall x \le y: \mu(x,y) = (-1)^{rk(y)-rk(x)}$, where $rk(\cdot)$ refers to the rank of elements in the lattice. For the proof of Theorem \ref{bpm_poly}, we only require the \textit{M\"obius numbers} of an Eulerian lattice. Thus, for completeness, we provide a simple proof of the identity regarding the M\"obius numbers of an Eulerian lattice, using the Euler-Poincar\'e Formula:

\begin{leftbar}
	\begin{lemma}
		\label{mobius_of_face_lattice}
		Let $Q$ be a polytope and denote by $F(Q)$ the set of all faces of $Q$. Let $\mathcal{P}=(F(Q) \cupdot \{\hat{0}\}, \le)$ be the face lattice of $Q$. The M\"obius numbers of $\mathcal{P}$ satisfy:
		\[ \forall x \in (F(Q) \cupdot \{\hat{0}\}):\ \mu_P(\hat{0},x) = (-1)^{rk(x)} \]
	\end{lemma}
\end{leftbar}
\begin{proof}
	Recall that every face of a polytope is also a polytope. Thus, for any face $x \in F(Q)$, we denote its face lattice by $\mathcal{P}_x$. The lattice $\mathcal{P}_x$ consists of all faces $y \in F(Q)$ where $y \le_P x$, thus $\mathcal{P}_x$ is a sub-lattice of $\mathcal{P}$ and $\mu_{\mathcal{P}}(\hat{0}, x) = \mu_{\mathcal{P}_x}(\hat{0}, x)$. By the definition of the face lattice, the rank of any face $y \in F(x)$ in $\mathcal{P}_x$ is given by $rk(x) = dim(x) + 1$, and thus agrees with its rank in $\mathcal{P}$. Consequently, we denote the rank of any face by $rk(\cdot)$.
	
	The proof proceeds by induction. If $x = \hat{0}$, the equality follows from the definition of the M\"obius function. Otherwise, let $x \in F(Q)$, where $k \eqdef rk(x) \ge 1$. By the definition of the M\"obius function and using the induction hypothesis:
	
	\[ \mu_{\mathcal{P}_x}(\hat{0}, x) =  -\sum_{\substack{y \in F(x) \cupdot \{\hat{0}\} \\ y \neq x}} \mu_{\mathcal{P}_x}(\hat{0}, y) = -\sum_{\substack{y \in F(x) \cupdot \{\hat{0}\} \\ y \neq x}} (-1)^{rk(y)} \]
	
	Since $x$ is a Polytope of dimension $k-1$,
	then by the Euler-Poincar\'e Formula for Polytopes (see, e.g., \cite{grunbaum2013convex}) we have:
	\begin{equation*}
	\begin{split}
	1 &= \sum_{j = 0}^{k-1} (-1)^j |\set{y \in F(x)}{dim(y) = j}| \\
	&= \sum_{x \ne y \in F(x)} (-1)^{rk(y)-1} + (-1)^{k-1} \\
	&= -\sum_{\substack{y \in F(x) \cupdot \{\hat{0}\} \\ y \neq x}} (-1)^{rk(y)} + (-1)^{k-1} + 1 = \mu_{\mathcal{P}_x}(\hat{0}, x) + (-1)^{k-1} + 1 \qedhere
	\end{split}
	\end{equation*}
	
\end{proof}

\defsubsubsection{Ranks in the Matching-Covered Lattice}

By Theorem \ref{elementary_faces_of_birkhoff}, the matching-covered lattice $\mathcal{P} = (MC_n \cupdot \{0\}, \subseteq)$ is isomorphic to the face lattice of the Birkhoff polytope, $B_n$. Thus, the lattice is \textit{ranked}, and its M\"obius numbers are given by $\mu_P(\hat{0}, x) = (-1)^{rk(x)}$. The following lemmas allow us to compute the rank of each matching-covered graph $G \in MC_n$, using its \textit{cyclomatic number} $\chi(G)$. 

\begin{leftbar}
	\begin{lemma}
		\label{order_of_subgraphs_of_elementary}
		Let $G$ be an elementary graph. The following inequality holds:
		\[ \forall G \ne H \in MC(G):\ \chi(H) < \chi(G) \]
		
	\end{lemma}
\end{leftbar}
\begin{proof}
	Let $G$ be an elementary graph and let $G \ne H \in MC(G)$. If $H$ is elementary, then $|E(H)| < |E(G)|$ and $|C(H)| = |C(G)| = 1$, thus $\chi(H) < \chi(G)$, as required. 
	
	Otherwise, observe that the connected components of $H$ are joined by edges in $G$, since $G$ is elementary and thus connected. Furthermore, we claim that every component of $H$ must be adjacent to at least 2 edges in $G$, one connected to a left vertex of the component, and another to a right vertex. Assume toward a contradiction that this is not the case, then there exists a component $C \in C(H)$ which is only adjacent to a single edge $e \in E(G)$. Since $G$ is elementary, there exists some perfect matching involving $e$ (i.e., the edge $e$ is allowed). However, upon selecting the edge $e$, the component $C$ becomes \textit{unbalanced}, and therefore the perfect matching cannot be extended over $C$, a contradiction.
	
	Thus, since each component of $G$ has at least two adjacent edges in $G$, we have that $|E(H)| + |C(H)| \le |E(G)|$ (i.e., if the adjacent edges form a cycle over the components $C(H)$). Thus:
	\begin{equation*}
		\begin{split}
			\chi(H) &= |E(H)| - |V(H)| + |C(H)| \le |E(G)| - |V(G|| < \chi(G) \qedhere 
		\end{split}
	\end{equation*}
\end{proof}
\begin{leftbar}
	\begin{corollary}
		\label{order_of_subgraphs_of_mu}
		Let $G \in MC_n$. The following inequality holds:
		\[ \forall G \ne H \in MC(G):\ \chi(H) < \chi(G) \]
	\end{corollary}
\end{leftbar}
\begin{proof}
	If $H$ is elementary, the proof follows from Lemma \ref{order_of_subgraphs_of_elementary}. Otherwise, the proof follows from the additivity of $\chi$, by applying Lemma \ref{order_of_subgraphs_of_elementary} to each connected component in which $G$ and $H$ differ.
\end{proof}
\begin{leftbar}
	\begin{lemma}
		\label{rank_of_mu_graph}
		Let $G \in MC_n$, $G \notin PM(K_{n,n})$. Then there exists $H \in MC(G)$ such that:
		
		\[\chi(H) = \chi(G) - 1\]
	\end{lemma}
\end{leftbar}
\begin{proof}
	Let $G \in MC_n$, $G \notin PM(K_{n,n})$. Since $G$ is not a perfect matching, there exists a component $C \in C(G)$ such that $C \ne K_2$. $C$ is elementary, and therefore there exists a bipartite ear decomposition: $C = e + P_1 + \dots + P_k$. Let $C' = e + P_1 + \dots P_{k-1}$, and observe that since $C'$ has a bipartite ear decomposition, it too is elementary. 
	
	If $P_k$ is a single edge, then we construct $H$ by taking $G$, and replacing the component $C$ with $C'$. Observe that $H \in MC_n$, and furthermore $|C(H)| = |C(G)|$ and $|E(H)| = |E(G)| - 1$. Thus $\chi(H) = \chi(G) - 1$, as required.
	
	Otherwise, $P_k$ is an ear $(v_1, u_1, \dots, v_t, u_t)$. In this case, we construct $H$ by taking $G$, replacing $C$ with $C'$, and replacing the ear $P_k$ with the edges $(v_2, u_1), \dots, (v_t, u_{t-1})$. Once again, $H \in MC_n$ (since all its components are elementary). Furthermore $|C(H)| = |C(G)| + t - 1$ and $|E(H)| = |E(G)| - t$, and thus $\chi(H) = \chi(G) - 1$.
\end{proof}

Thus, combining Corollary \ref{order_of_subgraphs_of_mu} and Lemma \ref{rank_of_mu_graph}, we find that:
\begin{leftbar}
	\begin{corollary}
		\label{rank_in_mu_lattice}
		Let $\mathcal{P} = (MC_n \cup \{\hat{0}\}, \subseteq)$ be the lattice of matching-covered graphs. Then:
		\[ \forall \hat{0} \ne G \in MC_n:\ \ rk(G) = \chi(G)+1 \]
	\end{corollary}
\end{leftbar}

\defsubsubsection{Completing the Proof of Theorem 1}

We are now ready to prove the main theorem for this section:

\begin{leftbar}
	\begin{customthm}{1}
	The unique multilinear polynomial representing $BPM_n$ over the Reals is:
		\[ BPM_{n}(x_{1,1}, \dots, x_{n,n}) = \sum_{G \in MC_n} {(-1)^{\chi(G)} \cdot \prod_{(i,j) \in E(G)} x_{i, j}}\]
	\end{customthm}
\end{leftbar}
\begin{proof}
	Let $\mathcal{P}=(MC_n \cup \{\hat{0}\}, \subseteq)$ be the lattice of matching-covered graphs, and let $B_n$ be the Birkhoff Polytope. Since $BPM_n$ is a graph cover function for the set $PM(K_{n,n})$, then by Proposition \ref{graph_cover_polynomial_mobius} we have:
	
	\[ BPM_n(x_{1,1}, \dots, x_{n,n}) = \sum_{G \in MC_n} -\mu_P(\hat{0}, G) \cdot \prod_{(i,j) \in E(G)} x_{i, j} \]
	
	By Theorem \ref{elementary_faces_of_birkhoff}, $\mathcal{P}$ is isomorphic to the face lattice of $B_n$, and thus by Corollary \ref{rank_in_mu_lattice} and Lemma \ref{mobius_of_face_lattice}, we get:
	
	\[ \forall G \in MC_n:\ \mu_P(\hat{0}, G) = (-1)^{rk(G)} = (-1)^{\chi(G)+1} \qedhere \]
\end{proof}

\defsubsection{Another Technique for Evasiveness?}

The proof regarding the multilinear polynomial of $BPM_n$ could, perhaps, be viewed as another ``technique'' for evasiveness. Given a (not necessarily bipartite) graph cover function whose corresponding lattice is isomorphic to the face lattice of some polytope, we can conclude that the function has full polynomial degree over the Reals, and is thus evasive. In fact, such functions also exhibit full polynomial degree over $\mathbb{F}_2$, and are therefore evasive even for $XOR$ decision trees. Nevertheless, we are presently only aware of two such functions exhibiting an isomorphism between their lattice and the face lattice of a polytope -- the first being that of bipartite perfect matching and the Birkhoff polytope, and the second being the $OR_n$ function and the n-dimensional Hypercube.

Previously, Kahn, Saks and Sturtevant \cite{kahn1984topological} showed a topological approach for evasiveness of monotone graph properties. Given a monotone graph property $\mathcal{P}$, their technique considers the abstract simplicial complex formed by all sets in the complement of $\mathcal{P}$, and shows that if the aforementioned complex is not contractible, then the property is evasive.

These two techniques are incomparable. While the \cite{kahn1984topological} technique is much more widely applicable, it does not imply that monotone graph properties exhibit full polynomial degree, neither over $\mathbb{F}_2$ nor over the Reals (and indeed, many do not). Nevertheless, our approach for evasiveness appears useful only in cases where the M\"obius numbers of the corresponding lattice are ``easy'' to compute, e.g. when the lattice is isomorphic to the face lattice of a polytope. Therefore, this technique appears rather limited.
\defsubsection{Corollaries of Theorem 1}
\label{comb_corollaries_subsection}

\begin{leftbar}
	\begin{corollary}
		\label{num_terms_bpm}
		The number of monomials in $BPM_n$ is at least $2^{n^2} \left(1 - \frac{2 n^4}{2^n}\right)$.
	\end{corollary}
\end{leftbar}
\begin{proof}
	The bound follows immediately from Proposition \ref{number_of_mu_graphs} and Theorem \ref{bpm_poly}.
\end{proof}

\begin{leftbar}
\begin{corollary}
	\label{fulldeg}
	The degree of $BPM_n$ over the Reals, as well as over $\mathbb{F}_2$, is $n^2$.
\end{corollary}
\end{leftbar}

\begin{leftbar}
\begin{corollary}
	\label{xor_evasiveness} 
	$BPM_n$ is \textit{evasive}, even for $XOR$ decision trees: 
	\[ D^{XOR}(BPM_n) = n^2 \]
\end{corollary}
\end{leftbar}
\begin{proof}
	We show that for any Boolean function $f$, $D^{XOR}(f) \ge deg_2(f)$, where $deg_2(f)$ is the degree of the polynomial representing $f$ over $\mathbb{F}_2$. Thus in particular $D^{XOR}(BPM_n) = n^2$. Let $f: \{0,1\}^n \rightarrow \{0,1\}$ be a Boolean function and let $T$ be a $XOR$ decision tree computing $f$. Let $\mathcal{P}$ be the set of all root to 1-leaf paths in $T$. For any path $P \in \mathcal{P}$ we construct the indicator over the path, denoted $\mathbbm{1}_P(x_1, \dots, x_n)$, by taking the product over any parity along the path (taking the parity itself for any right turn, and adding 1 to the term for any left turn). Observe that $f(x_1, \dots, x_n) = \sum_{P \in \mathcal{P}} \mathbbm{1}_P(x_1, \dots, x_n)$, therefore $deg_2(f) \le \max_{P \in \mathcal{P}} deg_2(\mathbbm{1}_P(x_1, \dots, x_n)) \le depth(T)$, where the last inequality follows since parities over $\mathbb{F}_2$ are linear functionals.
\end{proof}

\begin{leftbar}
	\begin{corollary}
		\label{odd_number_of_pm_graphs}
		The number of balanced bipartite graphs of order $2n$ containing a perfect matching is odd. Furthermore, the number of matching-covered graphs of order $2n$ is also odd.
	\end{corollary}
\end{leftbar}
\begin{proof}
	For any Boolean function $f$, $|\set{x \in \{0,1\}^n}{f(x)=1}| \equiv 1 \pmod 2$ if and only if the polynomial representing $f$ over $\mathbb{F}_2$ has full degree. Thus the number of graphs containing a perfect matching is odd. Let $H \in MC_n$. Clearly $H$ has a perfect matching, therefore:
	\begin{equation*}
		\begin{split}
			1 = BPM_n(H) &= \sum_{G \in MC_n} {(-1)^{\chi(G)} \cdot \mathbbm{1}\{G \subseteq H\}} \\
			&= \sum_{G \in MC(H)} {(-1)^{\chi(G)}} \end{split}
	\end{equation*}
	In particular $K_{n,n} \in MC_n$, thus:
	\begin{equation*}
		\begin{split}
			1 \equiv BPM_n(K_{n,n}) \pmod 2 \equiv |MC_n| \pmod 2 \qedhere
		\end{split}
	\end{equation*}
\end{proof}

\defsubsubsection{Lower Bound for $AND$ Decision Trees}

Given a Boolean function $f: \{0,1\}^n \rightarrow \{0,1\}$, the unique multilinear polynomial representing $f$ can be used to deduce lower bounds on the query complexity of $f$ for $AND$ decision trees. 

\begin{leftbar}
	\begin{lemma}
		\label{and_complexity_lb}
		Let $f: \{0,1\}^n \rightarrow \{0,1\}$ be a boolean function. Then:
		\[D^{AND}(f) \ge \log_3(|mon(f)|)\]
	\end{lemma}
\end{leftbar}
\begin{proof}
	Let $T$ be an AND decision tree computing $f$ and denote $d = depth(T)$. Let $\mathcal{P}$ be the set of all root to 1-leaf paths in $T$. For any $P \in \mathcal{P}$, construct the indicator function for the path as follows: 
	
	\[ \mathbbm{1}_{P}(x_1, \dots, x_n) = \left(\prod_{\lnot AND(S) \in P} \left(1-\prod_{i \in S} x_i\right)\right) \left(\prod_{ AND(S) \in P} \left(\prod_{i \in S} x_i\right) \right) \]  
	
	Notice that the multilinear polynomial of each indicator function of a path $P$ making $k$ left turns has at most $2^k$ monomials. Furthermore, in a binary tree of depth $d$ there are at most ${d \choose k}$ paths making exactly $k$ left turns (i.e., by selecting the position in the path at which the left turns are made). Finally, observe that $f(x_1, \dots, x_n) = \sum_{P \in \mathcal{P}} \mathbbm{1}_{P}(x_1, \dots, x_n)$. Thus by the uniqueness of the multilinear polynomial representing $f$, we have:
	\begin{equation*}
	\begin{split}
	|mon(f)| \le \sum_{P \in \mathcal{P}} |mon(\mathbbm{1}_{P})| \le \sum_{k=0}^{d} {d \choose k} 2^k = 3^d \qedhere
	\end{split}
	\end{equation*}
\end{proof}

Applying the aforementioned lemma to $BPM_n$ (and recalling Corollary \ref{num_terms_bpm}), we obtain:

\begin{leftbar}
	\begin{corollary}
		\label{bpm_and_lb}
		The depth of any $AND$ decision tree computing $BPM_n$ is at least:
		\[D^{AND}(BPM_n) \ge (\log_3 2) \cdot n^2 + o_n(1)\]
	\end{corollary}
\end{leftbar}

We note that the same decision tree lower bound can in fact be derived for the $l_1$-norm of the coefficient vector of the multilinear polynomial, rather than the number of monomials. Generally, the $l_1$ norm provides a stronger lower bound, however for the case of $BPM_n$ this does not yield a better bound, since all its coefficients have a magnitude of exactly 1.

\defsubsection{The Fourier Spectrum of $BPM_n$}
\label{fourier_subsubsection}

In this section we briefly discuss another multilinear polynomial representing $BPM_n$ over the Reals -- the Fourier Expansion of $BPM_n$. Given a multilinear polynomial over the Reals representing a Boolean function $f$ in the $\{0,1\}$ basis, the polynomial can be ``converted'' into the Fourier expansion of $f$ by replacing each monomial $\prod_{i \in S} x_i$ with the indicator $\mathbbm{1}_S(x_1, \dots, x_n) = \prod_{i \in S} \frac{-x_i + 1}{2}$, and applying the transformation $x \mapsto -2x + 1$ to the output. Thus:

\begin{leftbar}
	\begin{lemma}
		\label{convert_to_fourier}
		Let $f: \{0,1\}^n \rightarrow \{0,1\}$ be a Boolean function represented by the Real multilinear polynomial:
		
		\[f(x_1, \dots, x_n) = \sum_{S \subseteq [n]} a_S \cdot \prod_{i \in S} x_i\]
		
		\noindent Then the Fourier expansion of $f$ is given by:
		
		\[ \hat{f}(x_1, \dots, x_n) = 1 + \sum_{S \subseteq [n]} \left((-1)^{|S|-1} \sum_{T \supseteq S} \frac{a_T}{2^{|T|-1}}\right) \cdot \prod_{i \in S} x_i \]
	\end{lemma}
\end{leftbar}

Combining Lemma \ref{convert_to_fourier} and Theorem \ref{bpm_poly}, we thus conclude that: 

\begin{leftbar}
\begin{corollary}
	The Fourier coefficients of $BPM_n$ are given by:
	\[ \forall \hat{0} \ne G \subseteq K_{n,n}:\ \ \widehat{BPM^G_n} = (-1)^{|E(G)|} \sum_{\substack{H \supseteq G\\H \in MC_n}} \frac{(-1)^{\chi(H)-1}}{2^{|E(H)|-1}} \]
\end{corollary}
\end{leftbar}

While the above expression might be difficult to compute in the general case, we will now see that for the asymptotic majority of graphs (all elementary graphs), the Fourier coefficient can be exactly computed.

\begin{leftbar}
	\begin{proposition}
		\label{fourier_coeffs_of_c2n}
		Let $n \in \mathbb{N}^+$ and let $G \subseteq K_{n,n}$ be an elementary graph. Then: \[\widehat{BPM^G_n} = 2^{-n^2 + 1}\]
	\end{proposition}
\end{leftbar}
\begin{proof}
	If $G$ is elementary, then any graph $H \supseteq G$ is also elementary, as its ear decomposition is that of $G$, followed by adding single-edge ears for each edge in $E(H) \setminus E(G)$. Thus by Theorem \ref{bpm_poly} and Lemma \ref{convert_to_fourier}:
	\begin{equation*}
	\begin{split}
	\widehat{BPM^G_n} &= (-1)^{|E(G)| - 1} \sum_{\substack{H \supseteq G\\ H \in MC_n}} \frac{(-1)^{\chi(H)}}{2^{|E(H)| - 1}} \\
	&= \sum_{t=0}^{n^2 - |E(G)|} {n^2 - |E(G)| \choose t} \frac{(-1)^t}{2^{t + |E(G)| - 1}} = 2^{-n^2 + 1} \qedhere
	\end{split}
	\end{equation*}
\end{proof}

\begin{leftbar}
\begin{corollary}
	Let $n > 0$. For $(1 - o_n(1))\cdot 2^{n^2}$ of the Fourier coefficients, we have: 
	\[\widehat{BPM^G_n} = 2^{-n^2 + 1}\]
\end{corollary}
\end{leftbar}
\begin{proof}
	By Proposition \ref{number_of_mu_graphs}, the number of elementary graphs is at least:
	\[ |\set{G \subseteq K_{n,n}}{\text{G is elementary}}| \ge \left(1 - \frac{2n^4}{2^n}\right) \cdot 2^{n^2} = \left(1 - o_n(1)\right)\cdot 2^{n^2} \qedhere\]
\end{proof}

Lastly, we use the Fourier expansion of $BPM_n$ to derive a closed-form expression for the probability that a uniformly sampled random graph $G \sim G(n,n,\sfrac{1}{2})$ contains a bipartite perfect matching.

\begin{leftbar}
	\begin{proposition}
		Let $n > 0$. The probability that a perfect matching exists in a uniformly sampled balanced bipartite graph of order $2n$ is:
		
		\[ \Pr_{G \sim G\left(n,n,\sfrac{1}{2}\right)}\left[\text{G has a Perfect Matching}\right] = \sum_{G \in MC_n} \frac{(-1)^{\chi(G)}}{2^{|E(G)|}} \]
		
	\end{proposition}
\end{leftbar}
\begin{proof}
	By Theorem \ref{bpm_poly} and Lemma \ref{convert_to_fourier}, the Fourier coefficient of the empty set in $BPM_n$ is:
	
	\[ \widehat{BPM^\emptyset_n} = 1 - \sum_{G \in MC_n} \frac{(-1)^{\chi(G)}}{2^{|E(G)|-1}} \]
	
	\noindent Furthermore, for any Boolean function $f: \{1,-1\}^n \rightarrow \{1,-1\}$:
	\[ \hat{f}_\emptyset = \mathop{\mathbb{E}}_{x \sim \{1,-1 \}^n}\left[f(x)\right] = \mathop{\mathbb{E}}_{x \sim \{1,-1\}^n}\left[-2 \cdot \mathbbm{1}\{f(x)=-1\} + 1\right] = -2 \cdot \Pr_{x \sim \{1,-1\}^n}\left[f(x) = -1\right] + 1 \]
	
	\noindent And the equality now follows by rearranging.
\end{proof}

\defsection{The Dual Bipartite Perfect Matching Polynomial}
\label{dual_perfect_matching_polynomial_subsection}

In the previous {\sectionnamelower}, we dealt with the unique multilinear polynomial representing $BPM_n$ over the Reals in the $\{0,1\}$ basis ({\subsectionname} \ref{perfect_matching_polynomial_subsection}). We also briefly encountered the multilinear polynomial representing $BPM_n$ in the $\{1,-1\}$ basis, i.e., its Fourier expansion ({\subsectionname} \ref{fourier_subsubsection}).

We now turn our attention towards a \textit{third} multilinear polynomial, the one representing the ``dual function'' of $BPM_n$; namely, the function in which the symbols $0$ and $1$ have been ``flipped'', whereby $1$ indicates $False$ and $0$ is $True$ (which we refer to as the ``$\{1,0\}$-basis''). In this {\sectionnamelower} we will prove Theorem \ref{dual_bpm_poly}, which exhibits a fine-grained characterization of the dual polynomial. To this end, let us now introduce several more useful definitions and notation for this {\sectionnamelower}.

\defsubsection{Definitions and Notation} 
\defsubsubsection{Dual Functions}
\begin{leftbar}
	\begin{definition}
		\label{dual_polynomial}
		Let $f:\ \{0,1\}^n \rightarrow \{0,1\}$ be a Boolean function. The \textbf{dual function} of $f$, denoted $f^{\star}:\ \{0,1\}^n \rightarrow \{0,1\}$, is defined by:
		\[ \forall x \in \{0,1\}^n:\ f^{\star}(x_1, \dots, x_n) = 1 - f(1 - x_1, \dots, 1 - x_n) \]
	\end{definition}
\end{leftbar}

Hereafter, we denote by $\mathbf{BPM_n^\star}$ the dual function of $BPM_n$. For any graph $G \subseteq K_{n,n}$, we denote its 
corresponding coefficient in the polynomial representing $BPM_n^\star$ by $\mathbf{a_G^\star}$. Under this notation, the polynomial representing $BPM_n^\star$ is given by:

\[ BPM_n^\star(x_{1,1}, \dots, x_{n,n}) = \sum_{G \subseteq K_{n,n}} a_G^\star \cdot \prod_{(i,j) \in E(G)} x_{i, j} \]

\defsubsubsection{Graphs}

Hall's Theorem states that a balanced bipartite graph $G$ has \textit{no perfect matching} if and only if there exists an anticlique over a total of $n+1$ vertices.  In this {\sectionnamelower}, we consider dual functions, wherein the input bits (and output bit) are flipped. Thus, it will be useful to consider the ``dual'' to the above condition: the set of all \textit{complete bipartite graphs} over a total of $n+1$ vertices. We will hereafter refer to any such graph as a ``Hall Violator'', and will use the following notation for these graphs and for graphs which are ``covered'' by them:

\begin{leftbar}
	\begin{notation}
		\label{hall_violations_and_covers}
		Let $n > 1$. The set of all ``Hall Violator'' graphs is defined as follows:
		\[ \mathbf{HV_n} = \set{K_{X,Y} \subseteq K_{n,n}}{|X|+|Y| = n+1}\]
		\noindent Where $K_{X,Y}$ is the complete bipartite graph whose edges are $X \times Y$, and the remaining vertices are isolated. 
	\end{notation}
\end{leftbar}
\begin{leftbar}
	\begin{notation}
		Let $n > 1$. The set of all graphs which are ``covered'' by Hall violators is denoted by $\mathbf{HVC_n}$, where for every $G \subseteq K_{n,n}$:
		\[G \in HVC_n \iff \exists S \subseteq HV_n: \bar{E}(S)=E(G)\]
	\end{notation}
\end{leftbar}

We also consider the following two families of graphs:

\begin{leftbar}
	\begin{definition}
		\label{totally_and_strictly_totally_ordered}
		Let $n > 1$. A bipartite graph $G \subseteq K_{n,n}$ is called \textbf{totally ordered} if there exists an ordering of its left vertices $\{a_1, \dots, a_n\}$, such that:
		
		\[ N_G(a_1) \supseteq N_G(a_2) \supseteq \dots \supseteq N_G(a_n) \]
		
		\noindent Similarly, $G$ is called \textbf{strictly totally ordered} if in fact:
		
		\[ N_G(a_1) \supsetneq N_G(a_2) \supsetneq \dots \supsetneq N_G(a_n) \supsetneq \emptyset \]
	\end{definition}
\end{leftbar}

\defsubsection{A Fine Grained Characterization of the Dual Polynomial}

In this \subsectionnamelower, we obtain a fine grained characterization of the multilinear polynomial representing $BPM_n^\star$. Unlike the multilinear polynomial of $BPM_n$, we do not provide an explicit closed form of this polynomial. Nevertheless, we obtain an asymptotically tight estimate of the number of monomials appearing in the dual polynomial. Our characterization is the following:

\begin{leftbar}
	\begin{customthm}{2}
		\label{finer_characterization_of_dual}
		Let $n > 1$ and let $BPM_n^\star$ be the dual function of $BPM_n$, represented by the following multilinear polynomial over the Reals:
		
		\begin{equation*}
			BPM^{\star}_{n}(x_{1,1}, \dots, x_{n,n}) = \sum_{G \subseteq K_{n,n}} a^{\star}_G \prod_{(i,j) \in E(G)} x_{i,j}
		\end{equation*}
		Then for every $G \subseteq K_{n,n}$, we have:
		\begin{itemize}
			\item If $G$ \textit{is not totally ordered}, then $a^\star_G = 0$.
			\item If $G$ \textit{is strictly totally ordered}, then $a^\star_G = (-1)^{n+1}$
		\end{itemize}
	
	\end{customthm}
\end{leftbar} 

\noindent For the remainder of this \subsectionnamelower, we will set about proving Theorem \ref{finer_characterization_of_dual}.

\defsubsubsection{$\mathbf{BPM_n^\star}$ as a Graph Cover Function}

Let $G \subseteq K_{n,n}$. By Hall's Theorem, $G$ has a perfect matching if and only if its complement does not have a biclique over $n+1$ vertices. Therefore, by the definition of the dual function, we have:
\begin{equation*}
\begin{split}
BPM^\star_n(G) &= \mathbbm{1}\{\bar{G} \text{ does not have a perfect matching}\} \\
&= \mathbbm{1}\{\bar{G} \text{ has an anticlique over a total of } n+1 \text{ vertices}\} \\
&= \mathbbm{1}\{G \text{ has a biclique over a total of } n+1 \text{ vertices}\} \\
&= \mathbbm{1}\{\exists H \in HV_n,\ H \subseteq G\}
\end{split}
\end{equation*}

Thus, $BPM^\star_n$ is a graph cover function over the set $HV_n$. In particular, by Proposition \ref{monomials_graph_cover_functions}, the only monomials appearing in the multilinear polynomial representing $BPM_n^\star$ are those corresponding to graphs $G \in HVC_n$.

This observation alone already restricts the possible graphs which may appear as monomials of $BPM_n^\star$. For example, it allows us to deduce that every $G \subseteq K_{n,n}$ with $a^\star_G \ne 0$ has a single non-trivial connected component, since every $H \in HV_n$ appearing in $G$ contributes a connected component with exactly $n+1$ vertices. Nevertheless, this restriction does not suffice for bounding the number of monomials of $BPM_n^\star$ (as is exemplified later, in {\subsubsectionname} \ref{tot_ord_necessary}). Thus we now turn to our second characterization.

\defsubsubsection{Using The Eulerian Matching-Covered Lattice}

The characterization of $BPM_n^\star$ as a graph cover function for the lattice of graphs covered by Hall violators allowed us to restrict the set of \textit{monomials} that may appear in its polynomial representation. To gain further headway, we now shift our attention back to the Eulerian matching-covered lattice. Ideally, it would be advantageous to take the ``neat'' representation of $BPM_n$ in terms of the matching-covered lattice, and ``convert'' it into a characterization of $BPM_n^\star$.

Given a multilinear polynomial over the Reals representing any Boolean function $f: \{0,1\}^n \rightarrow \{0,1\}$, the dual polynomial of $f$ can immediately be derived by negating the inputs and output of $f$. In particular, if the polynomial representing $f$ is given by: 

\[ f(x_1, \dots, x_n) = \sum_{S \subseteq [n]} a_{S} \cdot \prod_{i \in S} x_i \]

\noindent Then the dual polynomial of $f$ can be expressed as follows:
\begin{equation*}
\begin{split}
	f^{\star}(x_1, \dots, x_n) &= 1 - f(1 - x_1, \dots, 1 - x_n) = 1 - \sum_{T \subseteq [n]} a_{T} \cdot \prod_{i \in T} (1 - x_i) \\
	&= 1 - \sum_{T \subseteq [n]} a_{T} \left( \sum_{S \subseteq T} (-1)^{|S|} \prod_{i \in S} x_i \right) \\
	&= 1 + \sum_{S \subseteq [n]} (-1)^{|S|+1} \left( \sum_{T \supseteq S} a_T \right) \cdot \prod_{i \in S} x_i
\end{split}
\end{equation*}

Thus by applying the above to $BPM_n$ and using the characterization of Theorem \ref{bpm_poly}, we obtain the following:

\begin{leftbar}
	\begin{lemma}
		\label{upper_sum_mobius}
		Let $\mathcal{P} = (MC_n \cupdot \{\hat{0}\}, \subseteq)$ be the matching-covered lattice. Then for every nonempty $G \subseteq K_{n,n}$, we have:
		\begin{equation*}
			\begin{split}
				a_G^\star = (-1)^{|E(G)| + 1} \sum_{\substack{G \subseteq H \subseteq K_{n,n}\\H \in MC_n}} (-1)^{\chi(H)} = (-1)^{|E(G)|}  \sum_{\substack{G \subseteq H \subseteq K_{n,n}\\H \in MC_n}} \mu_P(\hat{0}, H) \\	
			\end{split}
		\end{equation*}
	\end{lemma}
\end{leftbar}

We now show the following powerful characterization, which leverages the properties of the M\"obius function of an Eulerian lattice:

\begin{leftbar}
	\begin{lemma}
		\label{mu_graphs_zero_coeff}
		Let $n > 1$. For all $G \in (MC_n \setminus \{K_{n,n}\})$, we have $a^\star_G = 0$.
	\end{lemma}
\end{leftbar}
\begin{proof}
	
	Let $G \in (MC_n \setminus \{K_{n,n}\})$ and let $\mathcal{P} = (MC_n \cup \{\hat{0}\}, \subseteq)$ be the Eulerian matching-covered lattice, where $\hat{0}$ is the empty graph. Since $\mathcal{P}$ is Eulerian, its M\"obius function is multiplicative, thus: $\forall H \in MC_n, H \supseteq G:\ \mu_{\mathcal{P}}(\hat{0}, H) = \mu_{\mathcal{P}}(\hat{0}, G) \cdot \mu_{\mathcal{P}}(G, H)$. Therefore by Definition \ref{mobius_func_poset} and Lemma \ref{upper_sum_mobius}, we have: 
	
	\begin{equation*}
	\begin{split}
	a^\star_G &= (-1)^{|E(G)|}  \sum_{\substack{G \subseteq H \subseteq K_{n,n}\\H \in MC_n}} \mu_P(\hat{0}, H) \\
	&= (-1)^{|E(G)|} \mu_{\mathcal{P}}(\hat{0}, G) \sum_{\substack{G \subseteq H \subseteq K_{n,n}\\H \in MC_n}} \mu_{\mathcal{P}}(G, H) = 0 \qedhere
	\end{split}
	\end{equation*}
\end{proof}

\defsubsubsection{Extending Beyond the Matching-Covered Lattice}

By using the properties of the M\"obius function of an Eulerian lattice, we were able to deduce that all matching-covered graphs (other than the complete bipartite graph), have a zero dual coefficient. While matching-covered graphs constitute the asymptotic majority of all balanced bipartite graphs, the previous observation is nevertheless insufficiently powerful to obtain our bound (indeed there are at least $2^{n^2 - 2n}$ graphs that are not matching-covered).

Subsequently, we now extend our characterization to graphs beyond the matching-covered lattice. To this end, we introduce the notion of ``umbrellas'' -- a set of matching-covered graphs that forms a ``basis'' for a given graph $G$, even when $G$ itself is not matching-covered.

\begin{leftbar}
	\begin{notation}
		\label{umbrella_notation}
		Let $n > 1$ and let $G \subseteq K_{n,n}$ be a graph. The \textbf{Umbrella of G}, $\mathcal{U}(G) \subseteq MC_n$, is the set of all minimal matching-covered graphs, with respect to containment, which contain $G$ as a subgraph. Formally:
		
		\[H \in \mathcal{U}(G) \iff (G \subseteq H \in MC_n) \wedge (\not\exists H' \in MC_n:\ G \subseteq H' \subset H)\]
		
	\end{notation}
\end{leftbar}

The umbrella of $G$ is an anti-chain in the matching-covered lattice. In particular, any matching-covered graph $H \in MC_n$ containing $G$ as a subgraph, also contains a graph from the umbrella of $G$. Using umbrellas we now show the following identity for general graphs (i.e., not necessarily matching-covered):

\begin{leftbar}
	\begin{lemma}
		\label{base_of_graph}
		Let $n > 1$ and let $G \subseteq K_{n,n}$ be a nonempty graph. Then:
		
		\[ a_G^\star = (-1)^{n+|E(G)|} \cdot \sum_{\substack{\emptyset \ne S \subseteq \mathcal{U}(G)\\\bar{E}(S) = K_{n,n}}} (-1)^{|S|+1} \]
	\end{lemma}
\end{leftbar}
\begin{proof}
	Let $n > 1$ and let $\mathcal{P} = (MC_n \cup \{\hat{0}\}, \subseteq)$ be the Eulerian matching-covered lattice, where $\hat{0}$ is the empty graph. Let $G \subseteq K_{n,n}$, where $G \ne \hat{0}$. For any $\emptyset \ne S \subseteq MC_n$, denote by $\bigvee S$ the \textit{join} of all graphs in $S$. Recall (Proposition \ref{graph_cover_lattice}) that in $\mathcal{P}$, the join $\bigvee S$ is the union of all graphs in $S$. By Lemma \ref{upper_sum_mobius}, we have: \[a_G^\star = (-1)^{|E(G)|} \cdot \sum_{\substack{G \subseteq H \subseteq K_{n,n}\\H \in MC_n}} \mu_{\mathcal{P}}(\hat{0}, H)\]
	\noindent Using the inclusion-exclusion principle on the umbrella of $G$, $\mathcal{U}(G)$, we obtain:
	\begin{equation*}
	\begin{split}
	\sum_{\substack{G \subseteq H \subseteq K_{n,n}\\H \in MC_n}} \mu_{\mathcal{P}}(\hat{0}, H) &= \sum_{\emptyset \ne S \subseteq \mathcal{U}(G)} (-1)^{|S|+1} \sum_{\substack{(\bigvee S) \subseteq H \subseteq K_{n,n}\\H \in MC_n}} \mu_{\mathcal{P}}(\hat{0}, H) \\
	&= \left(\sum_{\substack{\emptyset \ne S \subseteq \mathcal{U}(G)\\(\bigvee S) \subset K_{n,n}}} (-1)^{|S|+1} \sum_{\substack{(\bigvee S) \subseteq H \subseteq K_{n,n}\\H \in MC_n}} \mu_{\mathcal{P}}(\hat{0}, H) + \sum_{\substack{\emptyset \ne S \subseteq \mathcal{U}(G)\\(\bigvee S) = K_{n,n}}} (-1)^{|S|+1} \mu_{\mathcal{P}}(\hat{0}, K_{n,n})\right)
	\end{split}
	\end{equation*}
	
	Since $\mathcal{P}$ is Eulerian, the sum of M\"obius numbers in any \textbf{nontrivial closed interval} is zero (see Lemma \ref{mu_graphs_zero_coeff}). In particular, for any $S \subseteq \mathcal{U}(G)$ where $(\bigvee S) \ne K_{n,n}$, we have:
	\[\sum_{\substack{(\bigvee S) \subseteq H \subseteq K_{n,n}\\H \in MC_n}} \mu_{\mathcal{P}}(\hat{0}, H) = 0\]
	Therefore:
	\begin{equation*}
	\begin{split}
	a^\star_G &= (-1)^{|E(G)|} \cdot \sum_{\substack{\emptyset \ne S \subseteq \mathcal{U}(G)\\\bar{E}(S) = K_{n,n}}} (-1)^{|S|+1} (-1)^{\chi(K_{n,n}) + 1} \\
	&= (-1)^{n+|E(G)|} \cdot \sum_{\substack{\emptyset \ne S \subseteq \mathcal{U}(G)\\\bar{E}(S) = K_{n,n}}} (-1)^{|S|+1} \qedhere
	\end{split}
	\end{equation*}
\end{proof}

Given a graph $G \subseteq K_{n,n}$, we say that $G$ has an \textbf{Incomplete Umbrella} if $\bar{E}(\mathcal{U}(G)) \ne K_{n,n}$, i.e., there exists some edge which is not present in any of the graphs in the umbrella of $G$. Observe that by Lemma \ref{base_of_graph}, this is a sufficient condition for exhibiting a zero dual coefficient.

\begin{leftbar}
	\begin{corollary}
		\label{small_basis_zero_coeff}
		Let $n > 1$ and let $G \subseteq K_{n,n}$ be a nonempty graph. Then:
		
		\[ \bar{E}(\mathcal{U}(G)) \ne K_{n,n} \implies a^\star_G = 0 \]
		
		\noindent 
	\end{corollary}
\end{leftbar}
\IfStrEq{\memomode}{ON}{
	\pagebreak
}{} 
\defsubsubsection{Graphs with an Incomplete Umbrella}

\begin{leftbar}
	\begin{definition}
		\label{wildcard_edge_defn}
		Let $n > 1$ and let $G \subseteq K_{n,n}$ be a nonempty graph. An edge $(a,b) \notin E(G)$ is called a \textbf{Wildcard Edge} for $G$ if and only if:
		
		\[ \forall H \in MC_n,\ H \supseteq G \cupdot \{(a,b)\}:\ (H \setminus \{(a,b)\}) \in MC_n \]
		
	\end{definition}
\end{leftbar}

\begin{leftbar}
	\begin{lemma}
		\label{wildcard_zero_coeff}
		Let $n > 1$ and let $G \subseteq K_{n,n}$ be a nonempty graph. Then:
		
		\[ G \text{ has a wildcard edge} \implies G \text{ has an incomplete umbrella} \]
	\end{lemma}
\end{leftbar}
\begin{proof}
	Let $(a,b) \notin E(G)$ be a wildcard edge for $G$. We show that $(a,b) \notin \bar{E}(\mathcal{U}(G))$. Assume towards a contradiction that $(a,b) \in \bar{E}(\mathcal{U}(G))$, and let $H \in \mathcal{U}(G)$ be a graph such that $(a,b) \in E(H)$. Then by the definition of $(a,b)$, we have $H' = (H \setminus \{(a,b)\}) \in MC_n$, and furthermore $G \subseteq H' \subset H$, in contradiction to the fact that $H \in \mathcal{U}(G)$.
\end{proof}

Building upon wildcard edges, we now introduce the following (slightly weaker) sufficient condition:

\begin{leftbar}
	\begin{definition}
		\label{surplus_edge_definition}
		Let $n > 1$ and let $G \subseteq K_{n,n}$ be a nonempty graph. Denote by $A$ the set of left vertices of $G$. An edge $(a,b) \notin E(G)$ is called a \textbf{Surplus Edge} for $G$ if and only if:
		
		\[ \forall X \subset A,\ a \in X,\ b \notin N_G(X):\ |N_G(X)| > |X| \]
	\end{definition}
\end{leftbar}

The above can be seen as a strengthening of Hall's condition, in which we require that the condition holds with a \textit{positive surplus}. However, note that we only require the condition for a particular family of sets -- those in which $a$ is present, and $b$ is not in the neighbour set. Finally, we show that surplus edges are, in fact, wildcard edges.

\begin{leftbar}
	\begin{lemma}
		\label{surplus_edges_are_wildcards}
		Let $n > 1$ and let $G \subseteq K_{n,n}$ be a nonempty graph. Then:
		\[ (a,b) \notin E(G) \text{ is a surplus edge for }G \implies (a,b) \notin E(G) \text{ is a wildcard edge for }G \]
		
	\end{lemma}
\end{leftbar}
\begin{proof}
	Let $(a,b) \notin E(G)$ be a surplus edge for $G$ and let $H \in MC_n$ such that $H \supseteq G \cupdot \{(a,b)\}$. Denote $H' = H \setminus \{(a,b)\}$. It remains to show that $H' \in MC_n$. Assume towards a contradiction that $H' \notin MC_n$ and denote by $C=(A_C \cupdot B_C, E) \in C(H)$ the connected component of $H$ containing the edge $(a,b)$. First, note that $K_2 \ne (C \setminus \{(a,b)\}) \in C(H')$, since elementary graphs are $2-connected$. Since $C \setminus \{(a,b)\}$ is not elementary, then by Theorem \ref{hetyei_conditions} there exists $\emptyset \ne X \subset A_C \subseteq A$ such that $|N_{H'}(X)| \le |X|$.
	
	Observe that $a \in X$ and $b \notin N_{H'}(X)$. Otherwise, we have $N_{H'}(X) = N_{H}(X)$ and since $H \in MC_n$ then $C$ is elementary and thus $|N_{H'}(X)| = |N_H(X)| > |X|$, a contradiction. However, since $(a,b)$ is a surplus edge for $G$, then for all $X \subset A$ such that $a \in X$, $b \notin N_G(X)$, we have $|N_G(X)| > |X|$. In particular, since $H' \supseteq G$, then for our $X$ we have $a \in X$ and $b \notin N_G(X)$ and thus $|N_{H'}(X)| \ge |N_G(X)| > |X|$, in contradiction to the definition of $X$.
\end{proof}

\defsubsubsection{Non-Totally Ordered Graphs Have a Zero Coefficient}
\label{dual_poly_upperbound_section}

Recall that the only monomials which may appear in $BPM_n^\star$ 
are those corresponding to graphs $G \in HVC_n$. Combining this characterization with those obtained using the Eulerian matching-covered lattice, we get:

\begin{leftbar}
	\begin{lemma}
		\label{non_containment_ordered_have_surplus_edges}
		Let $n > 1$ and let $G \in HVC_n$. If $G$ is \textbf{not totally ordered}, then $G$ has a surplus edge.
	\end{lemma}
\end{leftbar}
\begin{proof}
	Let $A = \{a_1, \dots, a_n\}$, $B = \{b_1, \dots, b_n\}$ be two sets, and let $G =(A \cupdot B, E) \in HVC_n$, such that $G$ is not totally ordered. Thus, there exist two vertices $a_i, a_j \in A$ such that:
	\[N(a_i) \not\supseteq N(a_j) \ \ \land\ \  N(a_j) \not\supseteq N(a_i) \ \ \land\ \  |N(a_i)| \ge |N(a_j)| \]
	We will show that $\forall b_k \in (N(a_j) \setminus N(a_i)):$ $(a_i,b_k)$ is a surplus edge for $G$. Let $b_k \in N(a_j) \setminus N(a_i)$, $b_m \in N(a_i) \setminus N(a_j)$. Since $G \in HVC_n$, every edge of $G$ is covered by some graph $K \in HV_n$, and in particular so are $(a_i, b_m)$, $(a_j, b_k)$. Thus, there exist $X_{i,m}, X_{j,k} \subseteq A$, $Y_{i,m}, Y_{j,k} \subseteq B$ such that:
	\begin{equation*}
	\begin{split}
	|X_{i,m}| + |Y_{i,m}| &= n+1, \\
	(X_{i,m} \times Y_{i,m}) &\subseteq E(G)
	\end{split}
	\quad\quad
	\begin{split}
	|X_{j,k}| + |Y_{j,k}| &= n+1 \\
	(X_{j,k} \times Y_{j,k}) &\subseteq E(G)
	\end{split}
	\end{equation*}
	
	and furthermore, $a_i \in X_{i,m}$, $b_m \in Y_{i,m}$, $a_j \in X_{j,k}$ and $b_k \in Y_{j,k}$. Assume towards a contradiction that $(a_i, b_k)$ is not a surplus edge for $G$. Then, there exists $X \subset A$ such that $a_i \in X$, $b_k \notin N(X)$ and $|N(X)| \le |X|$.
	\begin{figure}[h!]
		\centering
		\captionsetup{justification=centering,margin=1cm}
		\includegraphics[width=5.5cm, keepaspectratio]{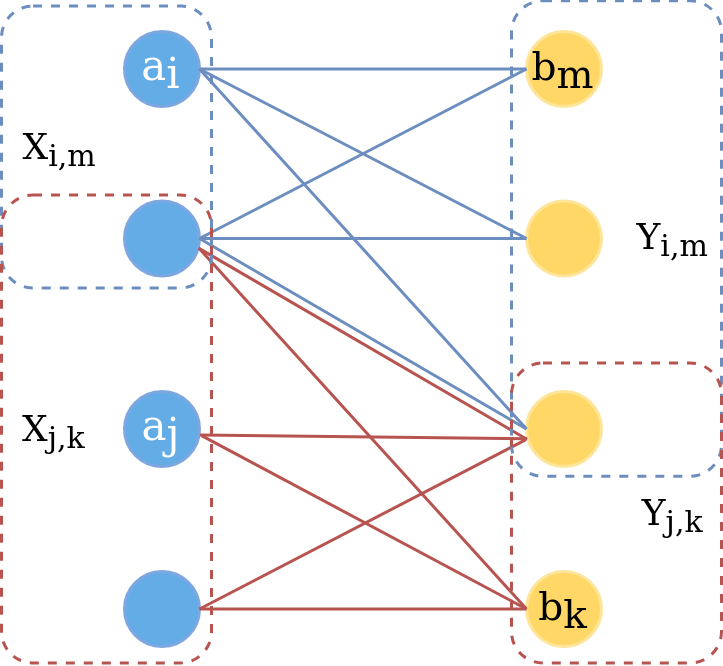}
		\caption{$G \in HVC_4$, which is not totally ordered. \newline The edge $(a_i, b_m)$ is covered by $K_{X_{i,m}, Y_{i,m}}$, and $(a_j, b_k)$ is covered by $K_{X_{j,k}, Y_{j,k}}$.} 
	\end{figure}
	
	Since $a_i \in X$ then $N(X) \supseteq N(a_i)$ and in particular $|N(X)| \ge |N(a_i)|$. Furthermore, observe that $X \cap X_{j,k} = \emptyset$, since otherwise $b_k \in N(X)$, in contradiction to the definition of $X$. Thus, we have $n - |X_{j,k}| \ge |X|$. Moreover, recall that by the definition of $a_i$ and $a_j$, we have $|N(a_i)| \ge |N(a_j)|$. Since the edge $(a_j, b_k)$ is covered by $K_{X_{j,k}, Y_{j,k}}$, then $N(a_j) \supseteq Y_{j,k}$. Lastly, by the definition of $X$, $|N(X)| \le |X|$. Putting all the above inequalities together, we have:
	\begin{equation*}
	\begin{split}
	n - |X_{j,k}| \ge |X| \ge |N(X)| \ge |N(a_i)| \ge |N(a_j)| \ge |Y_{j,k}|
	\end{split}
	\end{equation*}
	
	Therefore $|X_{j,k}| + |Y_{j,k}| \le n$, in contradiction to the fact that $|X_{j,k}| + |Y_{j,k}| = n+1$.	
\end{proof}
\begin{leftbar}
	\begin{corollary}
		\label{containment_order_is_necessary}
		Let $n > 1$ and let $G \subseteq K_{n,n}$. If $G$ is \textbf{not totally ordered}, then $a_G^\star = 0$.
	\end{corollary}
\end{leftbar}

\defsubsubsection{Strictly Totally Ordered Graphs Have a Non-Zero Coefficient}
\label{dual_poly_lowerbound_section}

\begin{leftbar}
	\begin{lemma}
		\label{lower_bound_graph}
		Let $n > 1$ and let $G \subseteq K_{n,n}$ be \textbf{strictly totally ordered}. Then: 
		\[a^\star_G = (-1)^{n+1}\]
	\end{lemma}
\end{leftbar}
\begin{proof}
	Let $G=(A \cupdot B, E) \subseteq K_{n,n}$ be a graph, where $A=\{a_1, \dots, a_n\}$ and $B = \{b_1, \dots, b_n\}$. The edges of $G$ are given by: $\forall i \in [n]: N_G(a_i) = \{b_1, \dots, b_i\}$. Observe that $G$ is \textbf{strictly totally ordered}, since $N(a_n) \supsetneq N(a_{n-1}) \supsetneq \dots N(a_1) \supsetneq \emptyset$.
	
	For every $k \in [n]$, denote $A_k = \{a_k, \dots, a_n\} \subseteq A$, $B_k = \{b_1, \dots, b_k\} \subseteq B$. By the definition of $G$, $\forall k \in [n]$, $K_{A_k, B_k} \subseteq G$. We now show that for any $K_{X,Y} \in HV_n$ such that $K_{X,Y} \subseteq G$ and $|Y| = k$, we necessarily have $X = A_k$ and $Y = B_k$.
	
	Assume towards a contradiction this is not the case. Let $K_{X,Y} \in HV_n$ such that $K_{X,Y} \subseteq G$, $|Y|=k$ and $Y \ne B_k$ or $X \ne A_k$. If $Y = B_k$, then for any $a_i \in X$ where $a_i \notin A_k$ (i.e., $i < k$), the edge $(a_i, b_1) \notin E(G)$ -- a contradiction. Otherwise, let $Y \ne B_k$ and let $j > k$ the maximal index such that $y_j \in Y$. By the definition of $G$, $\bigcap_{b \in Y} N_G(b) = N_G(y_j) = A_j$. Since $K_{X,Y} \subseteq G$, then in particular $X \subseteq A_j$, and therefore $|Y| + |X| \le k + |A_j| = n$, in contradiction to the fact that $K_{X,Y} \in HV_n$.
	
	Thus, the only Hall violators appearing in $G$ are the set:
	
	\[ \mathcal{H} \eqdef \set{H \in HV_n}{H \subseteq G} = \set{K_{A_k, B_k}}{k \in [n]} \]
	
	The union of all graphs in $\mathcal{H}$ is exactly $G$, therefore $G \in HVC_n$. However, recall that $BPM_n^\star$ is a graph cover function for the set $HV_n$, and thus by arithmetizing the formula representing the function (recall Proposition \ref{monomials_graph_cover_functions}), we get that:
	
	\[ BPM_n^\star(x_{1,1}, \dots, x_{n,n}) =  \sum_{G \in HVC_n} \left(\sum_{\substack{\emptyset \ne S \subseteq HVC_n\\ \bar{E}(S)=E(G)}}(-1)^{|S|+1}\right) \prod_{(i,j) \in E(G)} x_{i, j} \]
	
	Observe that the \textit{only} set of graphs $S \subseteq HVC_n$ whose union is equal to $G$ is the set $\mathcal{H}$ itself, since by omitting any $K_{A_k, B_k}$ we will fail to cover the edge $(a_k, b_k) \in E(G)$. Thus $a^\star_G = (-1)^{|\mathcal{H}| + 1} = (-1)^{n+1}$, as required.
	
	Lastly, we observe that \textbf{any} strictly totally ordered graph $G' \subseteq K_{n,n}$ is equivalent, up to permutations over each bipartition, to $G$ (and therefore has the same coefficient). This equivalence can be achieved by sorting the vertices of each bipartition by the cardinality of their neighbour sets, where the left vertices are sorted in ascending order, and the right vertices in descending order. 
\end{proof}

\centerline{This concludes the proof of Theorem \ref{finer_characterization_of_dual}.}

\defsubsection{Counting the Monomials of $BPM^\star_n$}

Using Theorem \ref{finer_characterization_of_dual}, we now deduce the following asymptotically tight bound on the number of monomials appearing in $BPM^\star$.

\begin{leftbar}
	\begin{corollary}
		\label{counting_monomials_dual}
		Let $n > 1$. The number of monomials in $BPM_n^\star$ satisfies:
		
		\[ (n!)^2 \le |mon(BPM_n^\star)| < (n+2)^{2n+2} \]
		
		And in particular: 
		
		\[\log_2 \left(|mon(BPM_n^\star)| \right) = 2n \log_2 n + O(n) = \Theta(n \log n)\]
	\end{corollary}
\end{leftbar}
\begin{proof}
	Let $n > 1$. For the lower bound, let $G$ be a strictly totally ordered graph. By Lemma \ref{lower_bound_graph}, \textit{all} strictly totally ordered graphs, and in particlar $G$, have $a_G^\star = (-1)^{n+1}$. However, since no two right or left vertices of $G$ have the same set of neighbours, any pair of permutations over the left and right bipartitions yields a new strictly totally ordered graph $\tilde{G} \cong G$, thus completing the lower bound. 
	
	For the upper bound, let $U = \{u_1, \dots, u_{n+1}\}$, $V = \{v_1, \dots, v_{n+1}\}$ be two sets. Denote by $C_n$ the set of all graphs $G \subseteq K_{n,n}$ that are \textit{totally ordered}. We begin by showing that:
	\[ |C_n| = \sum_{k=1}^{n+1} \left((k-1)! \cdot \stirlingII{n+1}{k}\right)^2 \]
	
	Where the notation $\stirlingII{n}{k}$ refers to the Stirling number of the second kind. To prove the equality, let us explicitly construct the set $C_n$ as follows; for every $1 \le k \le n+1$, let:
	\begin{equation*}
	\begin{split}
	U = U_1 \cupdot U_2 \dots \cupdot U_k \quad\quad\quad V = V_1 \cupdot V_2 \dots \cupdot V_k
	\end{split}
	\end{equation*}
	
	be partitions of $U,V$, respectively, into $k$ non-empty subsets, where without loss of generality $u_{n+1} \in U_k$ and $v_{n+1} \in V_k$. Then, for every $\pi, \tau \in S_{k-1}$, consider the graph $G \in C_n$, whose edges are given by:
	\[ \forall i \in [k-1]:\ \forall u \in U_{\pi(i)}:\ N_G(u) = V_{\tau(1)} \cupdot \dots \cupdot V_{\tau(i)} \]
	Recall that the number of partitions of $n$ elements into $k$ non-empty subsets is given by $\stirlingII{n}{k}$, the Stirling number of the second kind. Thus by the above construction, the cardinality of the set $C_n$ satisfies: 
	\[ |C_n| = \sum_{k=1}^{n+1} \underbrace{\vphantom{\stirlingII{n+1}{k}^2}\left((k-1)!\right)^2}_{\text{Choosing }\pi, \tau} \ \cdot \underbrace{\stirlingII{n+1}{k}^2}_{\text{Partitioning }U,V} \]
	Therefore:
	\begin{equation*}
	\begin{split}
	|mon(BPM_n^\star)| \le |C_n| &= \sum_{k=1}^{n+1} \left((k-1)! \cdot \stirlingII{n+1}{k}\right)^2  \\
	&\le \left(\sum_{k=1}^{n+1} k! \cdot \stirlingII{n+1}{k}\right)^2 = (F_{n+1})^2  
	\end{split}
	\end{equation*}
	
	Where $F_n$ denotes the n'th Fubini number. We now use the upper bound \cite{mezo2019combinatorics}: $\forall n \ge 1:\ F_n < (n+1)^n$, thereby concluding the proof.
\end{proof}

\defsubsubsection{Is the Totally Ordered Condition Necessary?}
\label{tot_ord_necessary}

Since $BPM_n^\star$ is a graph cover function for $HV_n$, the only monomials which \textit{may} appear in $BPM_n^\star$ are those corresponding to graphs $G \in HVC_n$ -- i.e., graphs covered by Hall violators. Clearly the number of Hall violators is $\Omega(2^{2n})$, however, one might wonder about a corresponding \textit{upper bound} for the number of graphs \textit{covered} by Hall violators. In particular, could we perhaps have derived as strong an asymptotic bound as the one yielded by the \textit{totally ordered} condition (Definition \ref{totally_and_strictly_totally_ordered}), by simply bounding the size of the set $HVC_n$? The following proposition shows that this is not the case, namely, there are (asymptotically) many more graphs which are \textit{covered} by Hall violators:

\begin{leftbar}
	\begin{proposition}
		\label{lowerbound_hvc_n}
		Let $n > 1$. Then:
		
		\[ \log_2(|HVC_n|) \ge \left\lfloor \frac{n}{2} \right\rfloor \left(\left\lceil \frac{n}{2} \right\rceil + 1\right) \ge \frac{n^2}{4} - 1 \]
	\end{proposition}
\end{leftbar}
\begin{proof}
	Let $n > 1$ and without loss of generality assume that $n = 2k$ where $k \in \mathbb{N}^+$. Let $A$,$B$ be two sets such that $|A|=|B|=n$. The lower bound follows by constructing a graph $G=(A \cupdot B, E_G) \in HVC_n$ where $|E(G)| = n^2 - n/2(n/2 + 1)$, such that $\{H \supseteq G\} \subseteq HVC_n$. First, partition each bipartition $A$,$B$ into two sets, as follows:
	\begin{equation*}
	\begin{split}
	A = (X \cupdot Y): \\
	B = (U \cupdot V): \\
	\end{split}
	\quad
	\begin{split}
	X &= \{a_1, \dots, a_{k}\}\ \ \ , \\
	U &= \{b_1, \dots, b_{k-1}\}\ ,
	\end{split}
	\quad
	\begin{split}
	Y &= \{a_{k + 1}, \dots, a_{2k}\} \\
	V &= \{b_{k}, \dots, b_{2k}\}
	\end{split}
	\end{equation*}
	
	The edges of $G$ are formed by connecting all edges between $X$ and $B$, and all edges between $Y$ and $U$, thus: $E(G) = (X \times B) \cup (Y \times U)$. Observe that $G \in HVC_n$, since it can be covered by taking $k$ copies of $K_{x,B}$, one for each $x \in X$, and taking another $k-1$ copies of $K_{A,u}$, one for each $u \in U$.
	
	Any missing edge $(y,v) \notin E(G)$ (where $y \in Y$ and $v \in V$) can be covered by $K_{X \cupdot \{y\}, U \cupdot \{v\}}$ (the complete bipartite graph connecting $X \cupdot \{y\}$ and $U \cupdot \{v\}$). Observe that $K_{X \cupdot \{y\}, U \cupdot \{v\}} \in HV_n$, since $|X \cupdot U \cupdot \{y, v\} | = n+1$. Thus $\{H \supseteq G\} \subseteq HVC_n$, as required.
\end{proof}

\defsubsection{Corollaries of Theorem 2}

\defsubsubsection{Communication Matrix Rank}

\paragraph{The 2-Player Communication Problem of Bipartite Perfect Matching} Consider the following communication problem. Given an input graph $G \subseteq K_{n,n}$, its edges are distributed between two players, Alice and Bob, according to some arbitrary fixed partition. The players' task is to devise a communication protocol to determine whether $G$ contains a bipartite perfect matching. Clearly, the rank of the associated communication matrix for this problem is at least exponential in $n$ (e.g., using a fooling set argument). Interestingly, the compact representation of $BPM_n^\star$ given by Theorem \ref{finer_characterization_of_dual} allows us to deduce that the rank of the communication matrix is, in fact, \textit{at most} exponential in $n \log n$.

\begin{corollary}
	Let $M$ be the communication matrix for the 2-player communication problem of bipartite perfect matching. The rank of $M$ over the Reals is bounded by:
	
	\[  rank_{\mathbb{R}}(M) \le (n+2)^{2n + 2} = 2^{\mathcal{O}(n \log n)} \]
\end{corollary}
\begin{proof}
	Let $M$ be the aforementioned communication matrix, and let $\bar{M} = J - M$, where $J$ is the all-ones matrix. The polynomial $BPM^\star_n$ induces an (at most) $|mon(BPM^\star_n)|$-rank decomposition of $\bar{M}$, since each monomial is a rank-1 matrix (see, e.g., \cite{nisan1995rank}). However, $rank(M) \le rank(\bar{M}) + 1$, and by Corollary \ref{counting_monomials_dual}, $|mon(BPM^\star_n)| < (n+2)^{2n + 2}$, thus concluding the proof.
\end{proof}

\defsubsubsection{Lower Bound for $OR$ Decision Tress}

Much in the same way that the multilinear polynomial representing $BPM_n$ allowed us to derive query complexity lower bounds for $AND$ decision trees, the multilinear polynomial representing $BPM_n^\star$ can be used to obtain similar lower bounds against $OR$ decision trees. The proof is very similar to that of Lemma \ref{and_complexity_lb}, but differs in several key steps, thus we provide it below for completeness.

\begin{leftbar}
	\begin{lemma}
		\label{or_complexity_lb}
		Let $f: \{0,1\}^n \rightarrow \{0,1\}$ be a boolean function. Then:
		\[D^{OR}(f) \ge \log_3(|mon(f^\star)|)\]
		
		\noindent Where $f^\star$ is the dual function of $f$.
	\end{lemma}
\end{leftbar}
\begin{proof}
	Let $T$ be an OR-decision tree computing $f$ and denote $d = depth(T)$. Let $\mathcal{P}$ be the set of all root to 0-leaf paths in $T$. For any $P \in \mathcal{P}$, the indicator function for the path is given by the following multilinear polynomial:
	
	\[ \mathbbm{1}_{P}(x_1, \dots, x_n) = \left(\prod_{\lnot OR(S) \in P} \left(1-\sum_{\emptyset \ne S \subseteq [n]} (-1)^{|S|+1} \prod_{i \in S} x_i\right)\right) \left(\prod_{ OR(S) \in P} \left(\sum_{\emptyset \ne S \subseteq [n]} (-1)^{|S|+1} \prod_{i \in S} x_i\right) \right) \]  
	
	By the definition of the dual function, we can construct $f^\star$ by summing the indicators for all paths $P \in \mathcal{P}$, where the inputs to each indicator are the negated input bits:
	\begin{equation*}
		\begin{split}
			f^\star(x_1, \dots, x_n) &= \sum_{P \in \mathcal{P}} \mathbbm{1}_{P}(1-x_1, \dots, 1-x_n) \\
			&= \sum_{P \in \mathcal{P}} \left(\prod_{OR(S) \in P} \left(1-\prod_{i \in S} x_i\right)\right) \left(\prod_{\lnot OR(S) \in P} \left(\prod_{i \in S} x_i\right) \right)
		\end{split}
	\end{equation*}
	
	Therefore, each path $P$ making $k$ right turns contributes at most $2^k$ monomials to $f^\star$. In a binary tree of depth $d$ there are at most ${d \choose k}$ paths making exactly $k$ right turns (i.e., by selecting the position in the path at which the right turns are made). Thus, we have:
	\begin{equation*}
	\begin{split}
	|mon(f^\star)| \le \sum_{k=0}^{d} {d \choose k} 2^k = 3^d \qedhere
	\end{split}
	\end{equation*}
\end{proof}

Applying the aforementioned lemma to $BPM_n$, we obtain:

\begin{leftbar}
	\begin{corollary}
		\label{bpm_or_lb}
		The depth of any $OR$ decision tree computing $BPM_n$ is at least:
		\begin{equation*}
		\begin{split}
		D^{OR}(BPM_n) &\ge 2 \log_3 (n!)
		\end{split}
		\end{equation*}
	\end{corollary}
\end{leftbar}

\defsubsection{Additional Coefficients of the Dual Polynomial}

Theorem \ref{finer_characterization_of_dual} offers a characterization of $BPM_n^\star$ in terms of \textit{totally ordered} and \textit{strictly totally ordered} graphs. The theorem states that only \textit{totally ordered} graphs may exhibit non-zero coefficients, and that all \textit{strictly totally ordered} indeed have non-zero coefficients. For graphs that are totally ordered but not strictly so, the situation is more complex \footnote{An additional analysis exclusively for graphs containing a perfect matching can be found in Appendix \ref{graph_with_pm_appendix}}. The following proposition shows that for any $n>2$, there exist graphs which are totally ordered but not strictly so, whose dual coefficient is $1$, $0$ and even $(n-2)^2$.

\begin{leftbar}
	\begin{proposition}
		Let $n>2$. There exist graphs $G \subseteq K_{n,n}$ which are totally ordered but not strictly so, such that:
		
		\[  \text{a.  }a^\star_G = 0,\quad\quad\quad \text{b.  }a^\star_G = 1,\quad\quad\quad \text{c.  }a^\star_G = (n-2)^2 \]
		
	\end{proposition}
\end{leftbar}
\begin{proof}
	Let $n > 2$ and let $A = \{a_1, \dots, a_n\}$, $B = \{b_1, \dots, b_n\}$ be two sets. Denote:
	
	 \[A_{n-1} = \{a_1, \dots, a_{n-1}\},\quad B_{n-1} = \{b_1, \dots, b_{n-1}\} \]
	 
	For the case $a^\star_G = 0$, consider any totally ordered graph such that $G \in (MC_n \setminus \{K_{n,n}\})$. For example, let $G = (A \cupdot B, E)$ such that $\forall i \in [n-1]:\ N_G(a_i) = B$ and $N_G(a_n) = \{b_1, b_2\}$. $G$ is totally ordered, since $N_G(a_1) \supseteq N_G(a_2) \supseteq \dots \supseteq N_G(a_n)$. However, we also have that $G \in MC_n$, therefore by Lemma \ref{mu_graphs_zero_coeff}, $a^\star_G = 0$.
	
	For the case $a^\star_G = 1$, consider any graph $G \in HV_n$. Observe that $G$ is both totally ordered and $G \in HVC_n$. Furthermore $G$ contains a \textit{single} Hall violator graph (itself), and is therefore a minterm of $BPM^\star_n$, and so $a^\star_G = 1$.
	
	Lastly, for the case $a^\star_G = (n-2)^2$, consider the graph $G = K_{A_{n-1}, B_{n-1}}$. Using Theorem \ref{hetyei_conditions}, the set of matching-covered graphs containing $G$, which we will denote by $\mathcal{H}$, can be partitioned into three sets $\mathcal{H} = \mathcal{H}_1 \cupdot \mathcal{H}_2 \cupdot \mathcal{H}_3$, as follows:
	\begin{equation*}
		\begin{split}
			\mathcal{H}_1 &= \{G \cupdot \{(a_n, b_n)\}\} \\
			\mathcal{H}_2 &= \set{G \cupdot \{(a_n,b_n)\} \cupdot (U \times \{b_n\}) \cupdot (\{a_n\} \times V)}{\emptyset \ne U \subseteq A_{n-1},\ \emptyset \ne V \subseteq B_{n-1}} \\
			\mathcal{H}_3 &= \set{G \cupdot (U \times \{b_n\}) \cupdot (\{a_n\} \times V)}{U \subseteq A_{n-1},\ V \subseteq B_{n-1},\ |U| \ge 2,\ |V| \ge 2}
		\end{split}
	\end{equation*}
	
	\noindent By Lemma \ref{upper_sum_mobius}, the dual coefficient of $G$ is given by:
	
	\begin{equation*}
		\begin{split} 
			a^\star_G &= (-1)^{|E(G)| + 1} \sum_{\substack{H \supseteq G\\H \in MC_n}} (-1)^{\chi(H)} = -\sum_{H \in \mathcal{H}} (-1)^{|E(H) \setminus E(G)| + |C(H)|} 
		\end{split}
	\end{equation*}
	
	\noindent For the single graph $H \in \mathcal{H}_1$, $|E(H) \setminus E(G)| = 1$, $|C(H)| = 2$, thus contributing $1$ to the sum. For each $H \in \mathcal{H}_2$, $|C(H)|=1$, thus $\mathcal{H}_2$'s contribution to the sum is:
	
	\[ \sum_{i=1}^{n-1} \sum_{j=1}^{n-1} {n-1 \choose i} {n-1 \choose j} (-1)^{i+j+1} = -1 \]
	 
	\noindent Lastly, for each $H \in \mathcal{H}_3$, $|C(H)|=1$. Thus $\mathcal{H}_3$'s contribution to the sum is:
	
	\[ \sum_{i=2}^{n-1} \sum_{j=2}^{n-1} {n-1 \choose i} {n-1 \choose j} (-1)^{i+j} = (n-2)^2 \]
	
	\noindent Summing up all the contributions, we get $a^\star_G = (n-2)^2$, thus concluding the proof.
	
\end{proof}

\section{Acknowledgments}

We thank Nati Linial for helpful discussions.

\bibliography{matching_polynomial}
\bibliographystyle{alpha}

\appendix
\pagebreak
\defsection{Graphs with a Perfect Matching}
\label{graph_with_pm_appendix}

In this section, we restrict our attention to graphs \textit{containing a perfect matching}, which appear in the dual polynomial $BPM_n^\star$. By Theorem \ref{dual_bpm_poly}, the only graphs appearing in the dual polynomial are those which are ``totally ordered''. However, by nature of having a perfect matching, a more precise characterization of their structure can be obtained.

Given a graph $G$ with a perfect matching, we consider the graph $G'$, formed by the union of \textit{all} perfect matchings of $G$. In this section, we show that if the monomial corresponding to $G$ appears $BPM^{\star}_n$, then the following conditions (and perhaps others) must hold. First, all the connected components of $G'$ must be complete bipartite graphs. Furthermore, for any edge in $G$ connecting two such components, all the edges between the components' corresponding bipartitions must appear. 


\begin{leftbar}
	\begin{lemma}
		\label{connected_component_not_complete_zero_coeff}
		Let $n>1$ and let $G \subseteq K_{n,n}$, where $G \notin MC_n$ and $PM(G) \ne \emptyset$. Denote by $G'$ the union of all the perfect matchings of $G$. If $G'$ has a connected component which is not a complete bipartite graph, then $a_G^\star = 0$.
	\end{lemma}
\end{leftbar}
\begin{proof}
	Let $G = (A \cupdot B, E) \notin MC_n$, where $PM(G) \ne \emptyset$ and denote by $G'$ the union of all perfect matchings of $G$. Let $C$ be a connected component of $G'$ which is not a complete bipartite graph. Let $(a,b) \in (A \cap V(C)) \times (B \cap V(C))$ be an edge such that $(a,b) \notin E(C)$. We will show that $(a,b)$ is a wildcard edge for $G$. Therefore, let $H \in MC_n$ be a graph such that $H \supseteq G \cupdot \{(a,b)\})$, and denote by $\tilde{H}$ the connected component of $H$ containing $C$.
	
	Observe that $\tilde{H} - V(C)$ contains a perfect matching (in particular, any of the perfect matchings induced by the components $C_i$ of $G'$ which are contained in $\tilde{H}$). Thus, $\tilde{H}$ has a bipartite ear decomposition of the form: $\tilde{H} = C + P_1 + \dots + P_q$, where there exists a path $P_i = (a,b)$ (since the vertices $a,b$ were present in $C$). Therefore $\tilde{H} \setminus \{(a,b)\}$ also has a bipartite ear decomposition: $\tilde{H} \setminus \{(a,b)\} = C + P_1 + \dots + P_{i-1} + P_{i+1} + \dots + P_q$, and by Theorem \ref{elementary_iff_ear_decomp}, $\tilde{H} \setminus \{(a,b)\}$ is elementary. Thus $(H \setminus \{(a,b)\}) \in MC_n$ and the proof follows by Lemma \ref{wildcard_zero_coeff}.
\end{proof}

\begin{leftbar}
	\begin{lemma}
		\label{complete_graphs_zero_coeff}
		Let $n > 1$ and let $G=(A \cupdot B, E) \subseteq K_{n,n}$. Denote $G'$ the union of $G$'s perfect matchings. If all the following conditions hold:
		
		\begin{enumerate}
			\item $G \in HVC_n$ and $PM(G) \ne \emptyset$.
			\item All the connected components of $G'$ are complete bipartite graphs.
			\item There exist $C_1 = (A_1 \cupdot B_1, E_1)$, $C_2 = (A_2 \cupdot B_2, E_2)$, where $C_1, C_2 \in C(G')$, such that:
			\[\emptyset \subsetneq \left((A_1 \times B_2) \cap E(G)\right) \subsetneq (A_1 \times B_2)\]
		\end{enumerate}
	
		\noindent Then $a_G^\star = 0$.
	\end{lemma}
\end{leftbar}
\begin{proof}
	Let $G$ be a graph satisfying the above conditions, and let $G'$, $C_1$, $C_2$ be the graphs described above. Denote $C(G') = \{C_1, \dots, C_t\}$, where $\forall i \in [t]:\ C_i = (A_i \cupdot B_i, E_i)$. Hereafter, we use the notation $C_i \leadsto C_j$ to denote an edge $(u,v) \in (A_i \times B_j)$.
	
	First, since $G \in HVC_n$ and $G$ has a perfect matching, then $G$ is connected. Let $(a,b),(u,v) \in (A_1 \times B_2)$ be two edges, such that $(a,b) \notin E(G)$ and $(u,v) \in E(G)$. We will show that $(a,b)$ is a wildcard edge of $G$. Let $H \in MC_n$ be a graph such that $H \supseteq G \cupdot \{(a,b)\})$. We will show that $H' = H  \setminus \{(a,b)\}$ is elementary, thus by Lemma \ref{wildcard_zero_coeff}, $a^\star_G = 0$. Let $(x,y) \in E(H')$. To show that $H'$ is elementary, by Theorem \ref{hetyei_conditions} it is sufficient to exhibit a perfect matching of $H'$ containing $(x,y)$. 
	
	Clearly, if $\exists i \in [t]: (x,y) \in E(C_i)$ then since $C_i$ is elementary, $C_i - x - y$ has a perfect matching, which can be extended to a perfect matching of $H'$ by adding a single perfect matching for each $C_j \in (C(G') \setminus C_i)$.
	
	Otherwise, denote by $C_i, C_j$ the components for which $x \in C_i$, $y \in C_j$. We begin by showing that $H$ has a directed cycle $\bar{C} = C_i \leadsto C_j \leadsto \dots \leadsto C_i$ containing $(x,y)$. Since $H \in MC_n$, every edge of $H$ participates in a perfect matching, and in particular so does $(x,y)$. Let $M$ be a perfect matching of $H$ involving $(x,y)$. Since $C_i - x$ is unbalanced, there must be some edge $C_k \leadsto C_i$ in $M$. Iteratively applying the same argument to $C_k$ and then to the component connected to it, we eventually gather a directed cycle $\bar{C} \in E(H)$ composed of edges of $M$, where $(x,y) \in \bar{C}$.
	
	Lastly, we use $\bar{C}$ to construct a perfect matching of $H'$ containing $(x,y)$. First, if $(a,b) \in \bar{C}$, then replace $(a,b)$ with $(u,v)$. Now, construct a perfect matching $\bar{M}$ as follows:
	
	\begin{enumerate}
		\item For each $C_k \notin \bar{C}$, take a single perfect matching over $C_k$.
		\item For each edge $(a_k, b_m) \in \bar{C}$, match $a_k$ and $b_m$.
		\item  For each $C_k \in \bar{C}$, denote $a_k \in A_k$, $b_k \in B_k$ the vertices of $C_k$ appearing in $\bar{C}$. By Theorem \ref{hetyei_conditions}, $C_k - a_k - b_k$ has at least one perfect matching (or is empty if $C_k = K_2$), which we add to $\bar{M}$.
	\end{enumerate}
\end{proof}

\pagebreak
\IfStrEq{\memomode}{ON}{
	\newgeometry{top=0.1cm, bottom=0.1cm, left=2cm}
}{} 
\defsection{The Polynomial of $BPM^\star_3$}
\label{bpm_star_3}
\begin{equation*}
	\begin{split}
		BPM^\star_3(x) &= x_{1, 1} x_{1, 2} x_{1, 3} + x_{1, 1} x_{2, 1} x_{3, 1} + x_{2, 1} x_{2, 2} x_{2, 3} + x_{1, 2} x_{2, 2} x_{3, 2} + x_{1, 3} x_{2, 3} x_{3, 3} + x_{3, 1} x_{3, 2} x_{3, 3} \\
		&+ x_{1, 1} x_{1, 2} x_{2, 1} x_{2, 2} + x_{1, 1} x_{1, 3} x_{2, 1} x_{2, 3} + x_{1, 2} x_{1, 3} x_{2, 2} x_{2, 3} + x_{1, 1} x_{1, 2} x_{3, 1} x_{3, 2} + x_{1, 1} x_{1, 3} x_{3, 1} x_{3, 3} \\
		&+ x_{1, 2} x_{1, 3} x_{3, 2} x_{3, 3} + x_{2, 1} x_{2, 2} x_{3, 1} x_{3, 2} + x_{2, 1} x_{2, 3} x_{3, 1} x_{3, 3} + x_{2, 2} x_{2, 3} x_{3, 2} x_{3, 3} - x_{1, 1} x_{1, 2} x_{1, 3} x_{2, 1} x_{2, 2} \\
		&- x_{1, 1} x_{1, 2} x_{1, 3} x_{2, 1} x_{2, 3} - x_{1, 1} x_{1, 2} x_{1, 3} x_{2, 2} x_{2, 3} - x_{1, 1} x_{1, 2} x_{1, 3} x_{2, 1} x_{3, 1} - x_{1, 1} x_{1, 2} x_{2, 1} x_{2, 2} x_{2, 3} \\
		&- x_{1, 1} x_{1, 2} x_{1, 3} x_{2, 2} x_{3, 2} - x_{1, 1} x_{1, 3} x_{2, 1} x_{2, 2} x_{2, 3} - x_{1, 2} x_{1, 3} x_{2, 1} x_{2, 2} x_{2, 3} - x_{1, 1} x_{1, 2} x_{1, 3} x_{2, 3} x_{3, 3} \\
		&- x_{1, 1} x_{1, 2} x_{2, 1} x_{2, 2} x_{3, 1} - x_{1, 1} x_{1, 2} x_{2, 1} x_{2, 2} x_{3, 2} - x_{1, 1} x_{1, 3} x_{2, 1} x_{2, 3} x_{3, 1} - x_{1, 1} x_{1, 2} x_{1, 3} x_{3, 1} x_{3, 2} \\
		&- x_{1, 1} x_{1, 2} x_{1, 3} x_{3, 1} x_{3, 3} - x_{1, 1} x_{1, 2} x_{1, 3} x_{3, 2} x_{3, 3} - x_{1, 1} x_{1, 3} x_{2, 1} x_{2, 3} x_{3, 3} - x_{1, 2} x_{1, 3} x_{2, 2} x_{2, 3} x_{3, 2} \\
		&- x_{1, 2} x_{1, 3} x_{2, 2} x_{2, 3} x_{3, 3} - x_{1, 1} x_{1, 2} x_{2, 1} x_{3, 1} x_{3, 2} - x_{1, 1} x_{1, 2} x_{2, 2} x_{3, 1} x_{3, 2} - x_{1, 1} x_{2, 1} x_{2, 2} x_{2, 3} x_{3, 1} \\
		&- x_{1, 1} x_{1, 3} x_{2, 1} x_{3, 1} x_{3, 3} - x_{1, 2} x_{2, 1} x_{2, 2} x_{2, 3} x_{3, 2} - x_{1, 1} x_{1, 3} x_{2, 3} x_{3, 1} x_{3, 3} - x_{1, 2} x_{1, 3} x_{2, 2} x_{3, 2} x_{3, 3} \\
		&- x_{1, 3} x_{2, 1} x_{2, 2} x_{2, 3} x_{3, 3} - x_{1, 2} x_{1, 3} x_{2, 3} x_{3, 2} x_{3, 3} - x_{1, 1} x_{2, 1} x_{2, 2} x_{3, 1} x_{3, 2} - x_{1, 2} x_{2, 1} x_{2, 2} x_{3, 1} x_{3, 2} \\
		&- x_{1, 1} x_{1, 2} x_{3, 1} x_{3, 2} x_{3, 3} - x_{1, 1} x_{2, 1} x_{2, 3} x_{3, 1} x_{3, 3} - x_{1, 1} x_{1, 3} x_{3, 1} x_{3, 2} x_{3, 3} - x_{1, 3} x_{2, 1} x_{2, 3} x_{3, 1} x_{3, 3} \\
		&- x_{1, 2} x_{1, 3} x_{3, 1} x_{3, 2} x_{3, 3} - x_{1, 2} x_{2, 2} x_{2, 3} x_{3, 2} x_{3, 3} - x_{1, 3} x_{2, 2} x_{2, 3} x_{3, 2} x_{3, 3} - x_{1, 1} x_{2, 1} x_{3, 1} x_{3, 2} x_{3, 3} \\
		&- x_{2, 1} x_{2, 2} x_{2, 3} x_{3, 1} x_{3, 2} - x_{1, 2} x_{2, 2} x_{3, 1} x_{3, 2} x_{3, 3} - x_{2, 1} x_{2, 2} x_{2, 3} x_{3, 1} x_{3, 3} - x_{2, 1} x_{2, 2} x_{2, 3} x_{3, 2} x_{3, 3} \\
		&- x_{1, 3} x_{2, 3} x_{3, 1} x_{3, 2} x_{3, 3} - x_{2, 1} x_{2, 2} x_{3, 1} x_{3, 2} x_{3, 3} - x_{2, 1} x_{2, 3} x_{3, 1} x_{3, 2} x_{3, 3} - x_{2, 2} x_{2, 3} x_{3, 1} x_{3, 2} x_{3, 3} \\
		&+ 2x_{1, 1} x_{1, 2} x_{1, 3} x_{2, 1} x_{2, 2} x_{2, 3} + x_{1, 1} x_{1, 2} x_{1, 3} x_{2, 1} x_{2, 2} x_{3, 1} + x_{1, 1} x_{1, 2} x_{1, 3} x_{2, 1} x_{2, 3} x_{3, 1} \\
		&+ x_{1, 1} x_{1, 2} x_{1, 3} x_{2, 1} x_{2, 2} x_{3, 2} + x_{1, 1} x_{1, 2} x_{1, 3} x_{2, 2} x_{2, 3} x_{3, 2} + x_{1, 1} x_{1, 2} x_{1, 3} x_{2, 1} x_{2, 3} x_{3, 3} \\
		&+ x_{1, 1} x_{1, 2} x_{1, 3} x_{2, 2} x_{2, 3} x_{3, 3} + x_{1, 1} x_{1, 2} x_{2, 1} x_{2, 2} x_{2, 3} x_{3, 1} + x_{1, 1} x_{1, 2} x_{1, 3} x_{2, 1} x_{3, 1} x_{3, 2} \\
		&+ x_{1, 1} x_{1, 3} x_{2, 1} x_{2, 2} x_{2, 3} x_{3, 1} + x_{1, 1} x_{1, 2} x_{1, 3} x_{2, 1} x_{3, 1} x_{3, 3} + x_{1, 1} x_{1, 2} x_{2, 1} x_{2, 2} x_{2, 3} x_{3, 2} \\
		&+ x_{1, 1} x_{1, 2} x_{1, 3} x_{2, 2} x_{3, 1} x_{3, 2} + x_{1, 1} x_{1, 2} x_{1, 3} x_{2, 3} x_{3, 1} x_{3, 3} + x_{1, 1} x_{1, 2} x_{1, 3} x_{2, 2} x_{3, 2} x_{3, 3} \\
		&+ x_{1, 1} x_{1, 3} x_{2, 1} x_{2, 2} x_{2, 3} x_{3, 3} + x_{1, 2} x_{1, 3} x_{2, 1} x_{2, 2} x_{2, 3} x_{3, 2} + x_{1, 2} x_{1, 3} x_{2, 1} x_{2, 2} x_{2, 3} x_{3, 3} \\
		&+ x_{1, 1} x_{1, 2} x_{1, 3} x_{2, 3} x_{3, 2} x_{3, 3} + 2x_{1, 1} x_{1, 2} x_{2, 1} x_{2, 2} x_{3, 1} x_{3, 2} + 2x_{1, 1} x_{1, 3} x_{2, 1} x_{2, 3} x_{3, 1} x_{3, 3} \\
		&+ 2x_{1, 1} x_{1, 2} x_{1, 3} x_{3, 1} x_{3, 2} x_{3, 3} + 2x_{1, 2} x_{1, 3} x_{2, 2} x_{2, 3} x_{3, 2} x_{3, 3} + x_{1, 1} x_{1, 2} x_{2, 1} x_{3, 1} x_{3, 2} x_{3, 3} \\
		&+ x_{1, 1} x_{2, 1} x_{2, 2} x_{2, 3} x_{3, 1} x_{3, 2} + x_{1, 1} x_{1, 3} x_{2, 1} x_{3, 1} x_{3, 2} x_{3, 3} + x_{1, 1} x_{1, 2} x_{2, 2} x_{3, 1} x_{3, 2} x_{3, 3} \\
		&+ x_{1, 2} x_{2, 1} x_{2, 2} x_{2, 3} x_{3, 1} x_{3, 2} + x_{1, 1} x_{2, 1} x_{2, 2} x_{2, 3} x_{3, 1} x_{3, 3} + x_{1, 2} x_{2, 1} x_{2, 2} x_{2, 3} x_{3, 2} x_{3, 3} \\
		&+ x_{1, 1} x_{1, 3} x_{2, 3} x_{3, 1} x_{3, 2} x_{3, 3} + x_{1, 3} x_{2, 1} x_{2, 2} x_{2, 3} x_{3, 1} x_{3, 3} + x_{1, 2} x_{1, 3} x_{2, 2} x_{3, 1} x_{3, 2} x_{3, 3} \\
		&+ x_{1, 3} x_{2, 1} x_{2, 2} x_{2, 3} x_{3, 2} x_{3, 3} + x_{1, 2} x_{1, 3} x_{2, 3} x_{3, 1} x_{3, 2} x_{3, 3} + x_{1, 1} x_{2, 1} x_{2, 2} x_{3, 1} x_{3, 2} x_{3, 3} \\
		&+ x_{1, 2} x_{2, 1} x_{2, 2} x_{3, 1} x_{3, 2} x_{3, 3} + x_{1, 1} x_{2, 1} x_{2, 3} x_{3, 1} x_{3, 2} x_{3, 3} + x_{1, 2} x_{2, 2} x_{2, 3} x_{3, 1} x_{3, 2} x_{3, 3} \\
		&+ x_{1, 3} x_{2, 1} x_{2, 3} x_{3, 1} x_{3, 2} x_{3, 3} + x_{1, 3} x_{2, 2} x_{2, 3} x_{3, 1} x_{3, 2} x_{3, 3} + 2x_{2, 1} x_{2, 2} x_{2, 3} x_{3, 1} x_{3, 2} x_{3, 3} \\
		&- x_{1, 1} x_{1, 2} x_{1, 3} x_{2, 1} x_{2, 2} x_{2, 3} x_{3, 1} - x_{1, 1} x_{1, 2} x_{1, 3} x_{2, 1} x_{2, 2} x_{2, 3} x_{3, 2} - x_{1, 1} x_{1, 2} x_{1, 3} x_{2, 1} x_{2, 2} x_{2, 3} x_{3, 3} \\
		&- x_{1, 1} x_{1, 2} x_{1, 3} x_{2, 1} x_{2, 2} x_{3, 1} x_{3, 2} - x_{1, 1} x_{1, 2} x_{1, 3} x_{2, 1} x_{2, 3} x_{3, 1} x_{3, 3} - x_{1, 1} x_{1, 2} x_{1, 3} x_{2, 2} x_{2, 3} x_{3, 2} x_{3, 3} \\
		&- x_{1, 1} x_{1, 2} x_{2, 1} x_{2, 2} x_{2, 3} x_{3, 1} x_{3, 2} - x_{1, 1} x_{1, 2} x_{1, 3} x_{2, 1} x_{3, 1} x_{3, 2} x_{3, 3} - x_{1, 1} x_{1, 2} x_{1, 3} x_{2, 2} x_{3, 1} x_{3, 2} x_{3, 3} \\
		&- x_{1, 1} x_{1, 3} x_{2, 1} x_{2, 2} x_{2, 3} x_{3, 1} x_{3, 3} - x_{1, 1} x_{1, 2} x_{1, 3} x_{2, 3} x_{3, 1} x_{3, 2} x_{3, 3} - x_{1, 2} x_{1, 3} x_{2, 1} x_{2, 2} x_{2, 3} x_{3, 2} x_{3, 3} \\
		&- x_{1, 1} x_{1, 2} x_{2, 1} x_{2, 2} x_{3, 1} x_{3, 2} x_{3, 3} - x_{1, 1} x_{1, 3} x_{2, 1} x_{2, 3} x_{3, 1} x_{3, 2} x_{3, 3} - x_{1, 2} x_{1, 3} x_{2, 2} x_{2, 3} x_{3, 1} x_{3, 2} x_{3, 3} \\
		&- x_{1, 1} x_{2, 1} x_{2, 2} x_{2, 3} x_{3, 1} x_{3, 2} x_{3, 3} - x_{1, 2} x_{2, 1} x_{2, 2} x_{2, 3} x_{3, 1} x_{3, 2} x_{3, 3} - x_{1, 3} x_{2, 1} x_{2, 2} x_{2, 3} x_{3, 1} x_{3, 2} x_{3, 3} \\
		&+ x_{1, 1} x_{1, 2} x_{1, 3} x_{2, 1} x_{2, 2} x_{2, 3} x_{3, 1} x_{3, 2} x_{3, 3} 
	\end{split}
\end{equation*}

\pagebreak
\IfStrEq{\memomode}{OFF}{
	\newgeometry{right=2cm, left=2cm}
}{} 
\defsection{The Monomials of $BPM_4^\star$}
\begin{figure}[h]
	\centering
	\captionsetup{justification=centering,margin=1cm}
	\includegraphics[width=18cm, keepaspectratio]{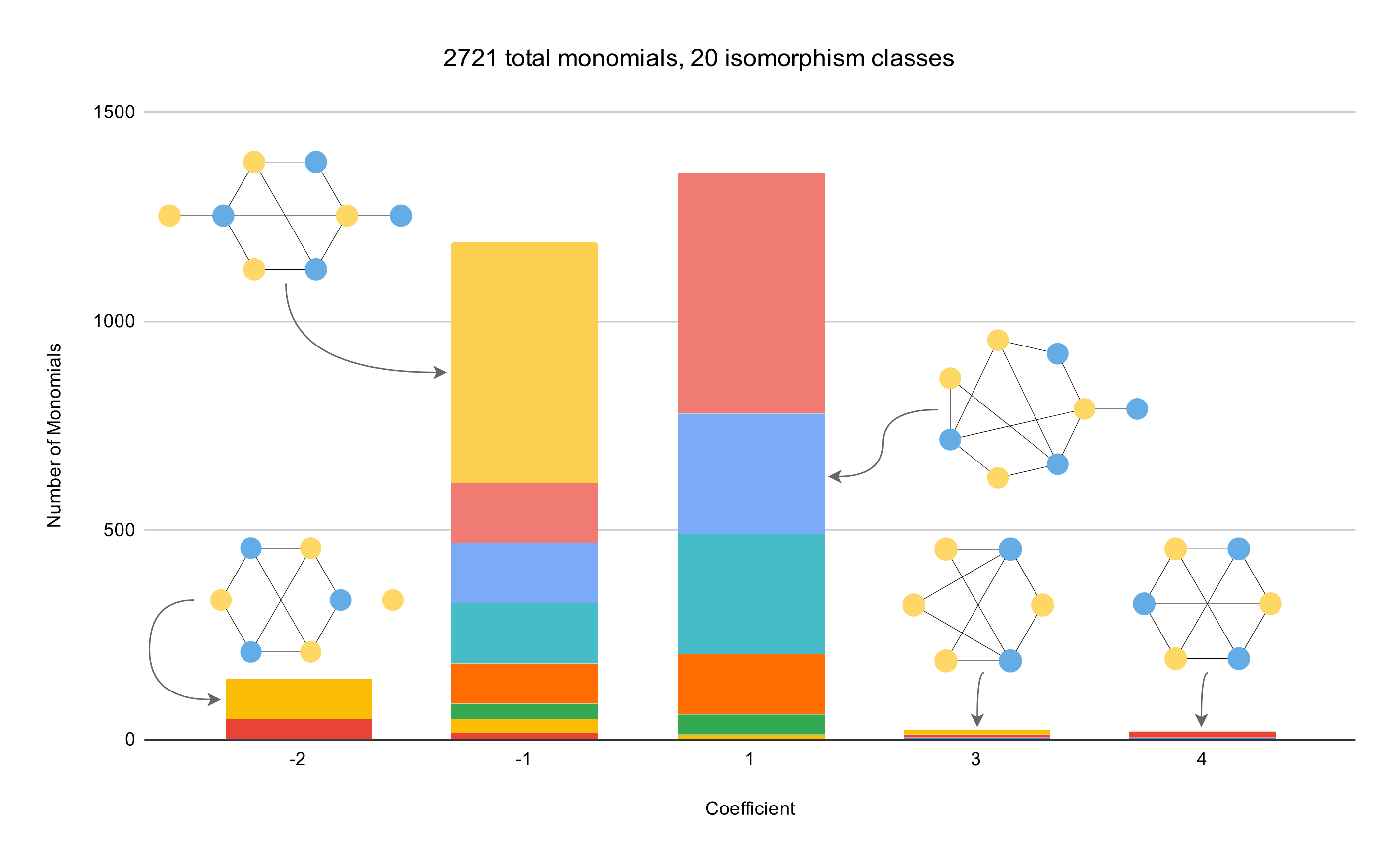}
	\caption{The monomials of $BPM_4^\star$, grouped by their coefficient. \newline For each coefficient, different colours indicate isomorphism classes.}
\end{figure}

\pagebreak
\newgeometry{left=0cm, top=1cm, bottom=2cm, right=0cm}
\begin{landscape}
\begin{figure}
	\parindent=10pt
	\defsection{The Matching-Covered Lattice for $n=3$}
	\includegraphics[width=25cm, keepaspectratio]{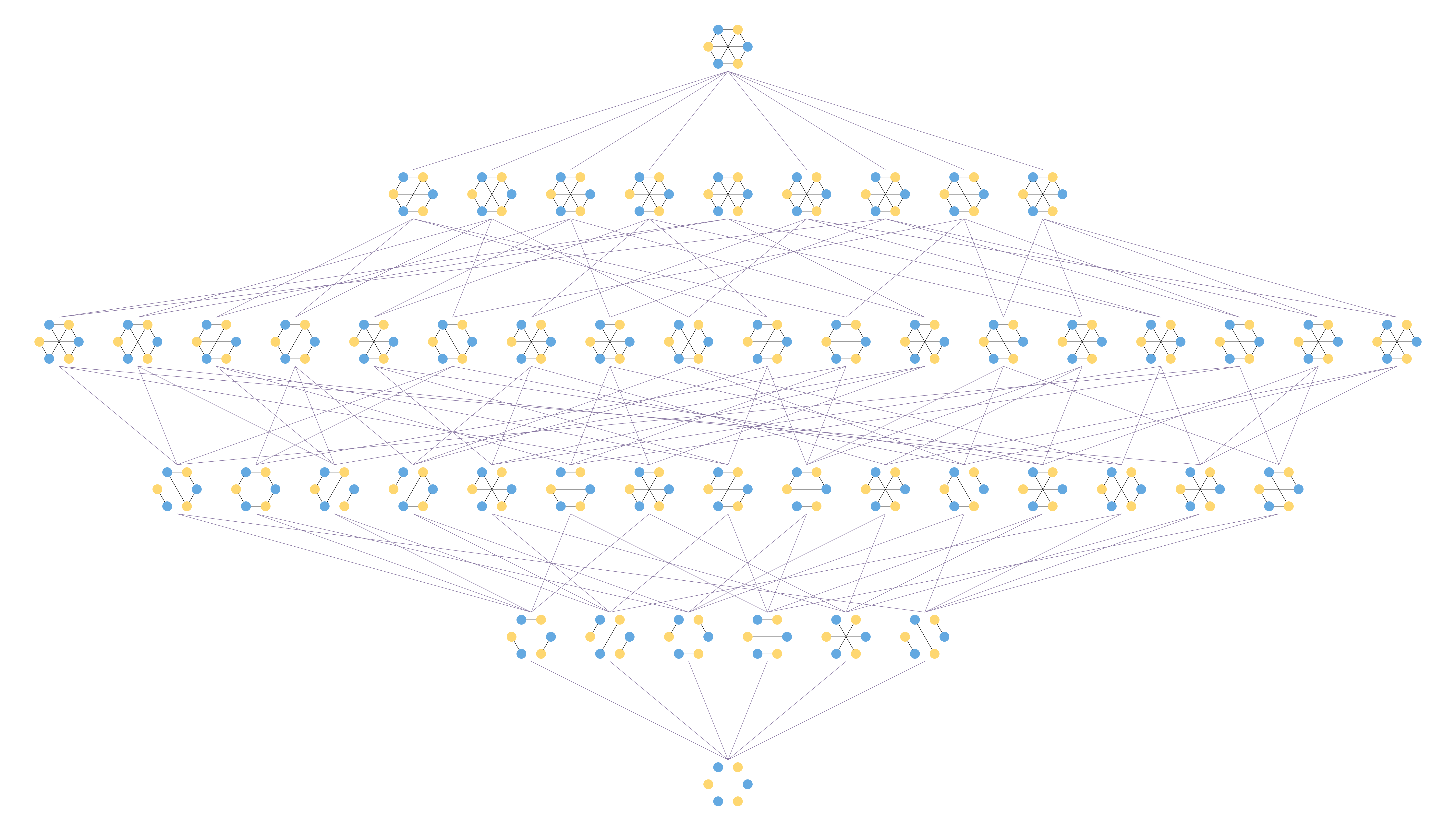}
	\caption{The Lattice $\mathcal{P}=(MC_3 \cupdot \{\hat{0}\}, \subseteq)$, which is isomorphic to the face lattice of the Birkhoff Polytope $B_3$}
	\label{fig:matching_covered_lattice_3}
\end{figure}
\end{landscape}

\end{document}